\documentclass[journal]{IEEEtran}
\usepackage{amsmath}
\usepackage{amsfonts}
\usepackage{algorithm}
\usepackage{algorithmic}
\usepackage{algorithm}
\usepackage{array}
\usepackage{textcomp}
\usepackage{stfloats}
\usepackage{url}
\usepackage{verbatim}
\usepackage{graphicx}
\usepackage{cite}
\usepackage{amsthm}
\usepackage[T1]{fontenc}
\usepackage{amssymb}
\usepackage{epsfig}
\usepackage{psfrag}
\usepackage{latexsym}
\usepackage{lettrine}
\usepackage{color}
\usepackage{multirow}
\usepackage{subfigure}

\usepackage{bm}
\usepackage{array}
\usepackage{hyperref}					
\hypersetup{colorlinks,
	linkcolor=blue,%
	anchorcolor=blue,
    citecolor=blue}

\newtheorem{thm}{Theorem}
\newtheorem{assumption}{Assumption}

\newtheorem{lemma}{Lemma}
\newtheorem{corollary}{Corollary}

\newtheorem{proposition}{Proposition}

\begin{document}
\title{Optimal Power Allocation for OFDM-based Ranging Using Random Communication Signals}
\author{
Ying Zhang,~\IEEEmembership{Graduate Student Member,~IEEE}, Fan Liu,~\IEEEmembership{Senior Member,~IEEE}, Tao Liu,~\IEEEmembership{Member,~IEEE},

and 
Shi Jin,~\IEEEmembership{Fellow,~IEEE}

\thanks{Y. Zhang and T. Liu are with the School of Automation and Intelligent Manufacturing, Southern University of Science
and Technology, Shenzhen 518055, China. (email: zhangying2024@mail.sustech.edu.cn; liut6@sustech.edu.cn).}
\thanks{F. Liu and S. Jin are with the National Mobile Communications Research Laboratory, Southeast University, Nanjing 210096, China. (email: f.liu@ieee.org; jinshi@seu.edu.cn).}

}

\maketitle
\begin{abstract} High-precision ranging plays a crucial role in future 6G Integrated Sensing and Communication (ISAC) systems. To improve the ranging performance while maximizing the resource utilization efficiency, future 6G ISAC networks have to reuse data payload signals for both communication and sensing, whose inherent randomness may deteriorate the ranging performance. To address this issue, this paper investigates the power allocation (PA) design for an OFDM-based ISAC system under random signaling, aiming to reduce the ranging sidelobe level of both periodic and aperiodic auto-correlation functions (P-ACF and A-ACF) of the ISAC signal. Towards that end, we first derive the closed-form expressions of the average squared P-ACF and A-ACF, and then propose to minimize the expectation of the integrated sidelobe level (EISL) under arbitrary constellation mapping. We then rigorously prove that the uniform PA scheme achieves the global minimum of the EISL for both P-ACF and A-ACF. As a step further, we show that this scheme also minimizes the P-ACF sidelobe level at every lag. Moreover, we extend our analysis to the P-ACF case with frequency-domain zero-padding, which is a typical approach to improve the ranging resolution. We reveal that there exists a tradeoff between sidelobe level and mainlobe width, and propose a project gradient descent algorithm to seek a locally optimal PA scheme that reduces the EISL. Finally, we validate our theoretical findings through extensive simulation results, confirming the effectiveness of the proposed PA methods in reducing the ranging sidelobe level for random OFDM signals.
\end{abstract}


\vspace{-0.5em}
\section{Introduction}
\IEEEPARstart{T}{he} International Telecommunication Union (ITU) has recently endorsed a global vision for 6G, identifying Integrated Sensing and Communication (ISAC) technology as one of six critical use cases for next-generation networks \cite{ITU2023}. To date, ISAC has been widely recognized as a key enabler for a number of emerging 6G applications, including autonomous vehicles, Internet of Things (IoT) networks, unmanned aerial vehicle (UAV) networks, and the low-altitude economy, which require simultaneous communication and sensing functionalities\cite{challengesfor6G2023,liufan2022integrated,Liufan2020IToC}. 

The ISAC system is a unified framework that provides both wireless communication and radar sensing functions. ISAC significantly enhances spectrum utilization and mitigates spectrum conflicts between radar and communication systems. Furthermore, it reduces the size and energy consumption of the equipment. Due to its high spectrum efficiency and low hardware costs, ISAC technology has garnered substantial attention from both academia and industry \cite{Wei2023,Wei2023TVT,7855671,8828023,zhang2021overview,Lishunyu2024TVT}. ISAC waveforms, as the basis of ISAC technology, have been studied extensively \cite{9924202,9724187,ZhangfanJSAC2024,Maotianqi2022TCOMM}. In general, waveform design methodologies for ISAC can be categorized into sensing-centric design, communication-centric design, and joint design approaches \cite{liufan2022integrated}. For the sensing-centric design, the priority of sensing functionality is higher than that of the communication counterpart. In such cases, one typically focuses on maximizing sensing performance or implementing basic communication capabilities over existing sensing waveforms or infrastructures\cite{R.M.M_1963},\cite{MR_2003}. Communication-centric schemes, on the other hand, rely on the existing communication waveforms and standard-compatible protocols, such as single-carrier (SC) signals, orthogonal time frequency space (OTFS) signals and orthogonal
frequency division multiplexing (OFDM) signals \cite{8424569},\cite{Sturm2011}. Different from the above two methods, the joint design strategy creates new ISAC waveforms, aiming to achieve an optimized balance between sensing and communication capabilities\cite{Liufan2018optimalwavedesign},\cite{liu-2021-cramer}.

The communication-centric approach is expected to be more advantageous for ISAC networks due to its lower implementation complexity\cite{Wei2023},\cite{Wei2023TVT}. Towards that end, numerous researches have been centering on communication-centric ISAC transmission, with a particular focus on waveform design and performance analysis \cite{sturm2011waveform, Chen2021, Zeng2020, Gaudio2020}. However, a critical challenge arises because communication data symbols must inherently exhibit \textit{randomness} to carry meaningful information, which can adversely affect sensing accuracy. This phenomenon has recently been characterized as the deterministic-random tradeoff (DRT), a fundamental limitation in ISAC system design \cite{xiong2023fundamental}. As a result, developing optimal communication-centric waveforms that mitigate sensing performance degradation becomes a crucial research objective. Against this backdrop, our recent work \cite{liu2024ofdmachieveslowestranging} took an initial step towards answering this fundamental question: \textit{What is the optimal communication-centric ISAC waveform under random signaling?} In \cite{liu2024ofdmachieveslowestranging}, we established the superiority of OFDM modulation in achieving the lowest ranging sidelobe. Specifically, we demonstrated that among all communication waveforms with a cyclic prefix (CP), OFDM is the only globally optimal waveform that achieves the lowest ranging sidelobe for both quadrature amplitude modulation (QAM) and phase shift keying (PSK) constellations. This holds true in terms of both the expectation of the integrated sidelobe level (EISL) and the expected sidelobe level at each lag of the periodic auto-correlation function (P-ACF). Furthermore, we proved that for communication waveforms without a CP, OFDM serves as a locally optimal waveform for QAM/PSK constellations, achieving a local minimum of the EISL of the aperiodic auto-correlation function (A-ACF).

Although the superiority of OFDM modulation over other waveforms in reducing ranging sidelobe level has been demonstrated in \cite{liu2024ofdmachieveslowestranging}, it generally assumed a uniform power distribution across transmitted symbols, leaving the optimal power allocation (PA) strategies for random ISAC signaling unexplored. Numerous studies have investigated PA schemes for OFDM-based ISAC systems \cite{Zhu2024powerallocation, YaoRubing_2023, Shi2019, Chen2023ICL, wangfangzhou_TSP}. In \cite{Zhu2024powerallocation}, the authors proposed three distinct PA schemes aimed at maximizing the mutual information (MI) of radar sensing under different scenarios. In \cite{YaoRubing_2023}, the authors designed an optimization algorithm to obtain the PA design over subcarriers, thus the sensing MI of the OFDM-based ISAC waveform is improved. The authors in \cite{Shi2019} conceived a power minimization-based joint subcarrier assignment and power allocation (PM-JSAPA) approach for ISAC systems. Similarly, a joint subcarrier and PA method for the integrated OFDM with interleaved subcarriers (OFDM-IS) waveform was proposed in \cite{Chen2023ICL}. The proposed method minimizes the autocorrelation sidelobe level of the integrated OFDM-IS waveform, using a cyclic minimization algorithm to solve the optimization problem. In \cite{wangfangzhou_TSP}, the authors proposed two designs for a radar-communication spectrum sharing problem by maximizing the output signal-to-interference-plus-noise ratio (SINR) at the radar receiver while maintaining certain communication throughput and power constraints. Although extensive studies on OFDM PA have been conducted for deterministic signaling, none have rigorously addressed the issue of PA design for random signaling.

In this paper, we investigate the PA design for random OFDM communication signals, in the context of target range estimation under a monostatic ISAC setup. In this setup, the ISAC transmitter (Tx) sends out an OFDM communication signal modulated with random symbols for dual purposes of delivering information to communication users and multi-target ranging. The signal is received by communication users, and simultaneously reflected from distant targets back to a sensing receiver (Rx) that is colocated with the ISAC Tx. As a result, despite its inherent randomness, the sensing Rx has full knowledge of the ISAC signal. To evaluate the ranging performance under matched filtering algorithms, we adopt the auto-correlation function (ACF) of the ISAC signal as a performance indicator, considering both periodic and aperiodic convolutions for signals with and without CP, respectively. For clarity, we summarize our main contributions as follows:

\begin{itemize}
\item [\textbullet] We develop a generic framework to analyze the ranging
performance of OFDM-based random ISAC signals. Specifically, we analyze P-ACF and A-ACF for the considered signal by deriving their closed-form expressions.

\item [\textbullet] We prove that the uniform PA scheme achieves the lowest EISL for both P-ACF and A-ACF. In particular, it also achieves the lowest average sidelobe level at every lag of the P-ACF.

\item[\textbullet] We derive a closed-form expression for P-ACF with PA and frequency zero-padding, where the latter is a classical method to improve the range resolution of the matched filtering output. We demonstrate that in the case of frequency zero-padding, uniform PA does not yield the lowest possible ranging sidelobe. To address this, we formulate an optimization problem and develop the Projected Gradient Descent (PGD) algorithm to conceive a PA scheme that minimizes the ranging sidelobe. 

\item[\textbullet] By comparing ranging performances, we observe that while the PA scheme obtained via the PGD algorithm achieves a lower sidelobe level, it results in a slight increase in the mainlobe width compared to the uniform PA scheme. Building upon this observation, we introduce a new constraint to limit the mainlobe width and employ the Successive Convex Approximation (SCA) algorithm to solve the corresponding optimization problem. Our results indicate the existence of a tradeoff between the mainlobe width and the sidelobe level.
\end{itemize}

The remainder of this paper is organized as follows. Section \ref{sec2} introduces the system model of the ISAC system under consideration and the corresponding performance metrics. Section \ref{Sec_III} derives the closed-form expressions for both P-ACF and A-ACF of random ISAC signals using OFDM modulation and PA schemes. The optimal PA strategies of P-ACF and A-ACF with OFDM modulation for random ISAC signals are discussed in Section \ref{Sec_IV}. Section \ref{Sec V} conceives optimal PA strategies for P-ACF with frequency zero-padding and OFDM modulation in random ISAC signals. Section \ref{Sec_VI} presents simulation results to validate the theoretical analysis. Finally, Section \ref{Sec_VII} concludes the paper.

{\emph{Notations}}: Matrices are denoted by bold uppercase letters (e.g., $\mathbf{U}$), vectors are represented by bold lowercase letters (e.g., $\mathbf{x}$), and scalars are denoted by normal font (e.g., $N$); The $n$th entry of a vector $\mathbf{s}$, and the $(m,n)$-th entry of a matrix $\mathbf{A}$ are denoted as $s_n$ and $a_{m.n}$, respectively; $\otimes$ and $\operatorname{vec}\left(\cdot\right)$ denote the Kronecker product and the vectorization, $\left(\cdot\right)^T$, $\left(\cdot\right)^H$, and $\left(\cdot\right)^*$ stand for transpose, Hermitian transpose, and the complex conjugate of the matrices;  $\ell_p$ norm is written as $\left\| \cdot\right\|_p$, and $\mathbb{E}(\cdot)$ represents the expectation operation; The notation Diag($\mathbf{a}$) denotes the diagonal matrix obtained by placing the entries of $\mathbf{a}$ on its main diagonal.

\section{System Model} \label{sec2}
\subsection{ISAC Signal Model}

This paper considers a monostatic ISAC system. The ISAC Tx emits a signal modulated with random communication symbols, which is received by a communication Rx and simultaneously reflected back to a sensing Rx by one or more targets at varying ranges. The sensing Rx, collocated with the ISAC Tx, performs matched filtering to estimate target delay parameters using the known random ISAC signal.

Let $\mathbf{s} = \left[s_1, s_2, \dots, s_N \right]^\mathrm{T} \in \mathbb{C}^{N}$ represent the $N$ communication symbols to be transmitted. When the OFDM modulation is adapted and the power allocation is considered, the discrete-time domain signal is given by
    \begin{equation}\label{Signal with power allocation}
        \mathbf{x} = \mathbf{F}_N^H\mathbf{B}\mathbf{s},
    \end{equation} 
where $\mathbf{F}_N$ is the normalized discrete Fourier transform (DFT) matrix of size $N$, and the matrix $\mathbf{B}$ is the PA matrix, in the form of
    \begin{equation}\label{B}
    \begin{aligned}
        \mathbf{B} = \mathrm{Diag} \left(\left[\sqrt{P_1},\sqrt{P_2},...,\sqrt{P_N}\right]^T\right),   
    \end{aligned}        
    \end{equation}
where $P_i$ represents the power allocated to $i$th subcarrier, with the constraints $\sum_{i=1}^N P_i = N$ and $P_i \ge 0$, for all $i$.

\begin{assumption}[Unit Power and Rotational Symmetry]
    We focus on constellations with a unit power, zero mean, and zero pseudo-variance, defined as
    \begin{equation}
        \mathbb{E}( \left|s_n \right|^2) = 1, \quad\mathbb{E}(s_n) = 0,\quad\mathbb{E}(s_n^2) = 0,\quad \forall n.
    \end{equation}
\end{assumption}

\noindent The unit-power normalization in Assumption 1 enables fair sensing and communication performance under varying constellation formats. Moreover, the properties of zero mean and zero pseudo-variance hold for most practical constellations (e.g., PSK and QAM families, excluding BPSK and 8-QAM).

To proceed, let us further define the \textit{kurtosis} of a constellation as
\begin{equation}
    \mu_4 = \frac{\mathbb{E}\left\{|s_n-\mathbb{E}(s_n)|^4\right\}}{\mathbb{E}\left\{|s_n-\mathbb{E}(s_n)|^2\right\}^2}.
\end{equation}
Under unit power and zero mean conditions, kurtosis reduces to the fourth-order moment $\mu_4= \mathbb{E}(|s_n|^4)$. Using the Power Mean Inequality, we arrive at
\begin{equation}
    \mathbb{E}(|s_n|^4) \geq \left( \mathbb{E}(|s_n|^2) \right)^2 = 1.
\end{equation}
Specifically, all PSK constellations exhibit $\mu_4 = 1$, while all QAM constellations have $1 \leq \mu_4 \leq 2$.

\subsection{Sensing Performance Metric}

In radar systems, the Ambiguity Function (AF) of the signal is an important performance indicator for sensing \cite{Levanon2013RADARS}. 
    \begin{equation}\label{self-AF}
        \mathcal{X}(\tau,f_d) = \int_{-\infty}^{\infty} x(t) x^*(t-\tau) e^{j 2 \pi f_d t} dt ,
    \end{equation}
where $x(t)$ is the radar signal, $\tau$ is the time delay and $f_d$ is Doppler shift. For the sake of simplicity, this paper considers only the zero-Doppler slice and sets aside the characterization of the complete AF as our future work. By letting $f_d = 0$, \eqref{self-AF} can be rewritten as
    \begin{equation}\label{self-AF-2}
        \mathcal{X}(\tau,0) = \int_{-\infty}^{\infty} x(t) x^*(t-\tau) dt ,
    \end{equation}
which is known as the ACF of the signal $x(t)$, and is well-recognized as a key performance indicator for ranging. In a discrete form, the ACF can be defined as the linear or periodic self-convolution of the discrete signal $\mathbf{x}$, depending on whether a CP is added.

\subsubsection{Aperiodic ACF (A-ACF)}
\begin{equation}
    r_k = \mathbf{x}^H\mathbf{J}_k\mathbf{x} = r_{-k}^*, \quad k = 0,1,\ldots, N-1,
\end{equation}
where $\mathbf{J}_k$ is the $k$th shift matrix in the form of
\begin{equation}
    \mathbf{J}_k = {\begin{bmatrix}
   {\mathbf{0}}&{{{\mathbf{I}}_{N - k}}} \\ 
  {\mathbf{0}}&{\mathbf{0}}        
    \end{bmatrix}}
, \quad
    \mathbf{J}_{-k} = \mathbf{J}_k^T=\begin{bmatrix}
   {\mathbf{0}}&{\mathbf{0}} \\ 
  {{{\mathbf{I}}_{N - k}}}&{\mathbf{0}}        
    \end{bmatrix}.
\end{equation}

\subsubsection{Periodic ACF (P-ACF)}
\begin{equation}  \label{P-ACF self acco}
    \widetilde{r}_k = \mathbf{x}^H\widetilde{\mathbf{J}}_k\mathbf{x} = \widetilde{r}_{N-k}^*, \quad k = 0,1,\ldots, N-1,
\end{equation}
where $\widetilde{\mathbf{J}}_k$ is defined as the $k$th periodic shift matrix \cite{Stoica2009}, given as
\begin{equation}
    \widetilde{\mathbf{J}}_k = \begin{bmatrix}
  {\mathbf{0}}&{{{\mathbf{I}}_{N - k}}} \\ 
  {\mathbf{I}_k}&{\mathbf{0}}         
    \end{bmatrix}, \quad
    \widetilde{\mathbf{J}}_{-k} = \widetilde{\mathbf{J}}_{N-k} = \begin{bmatrix}
   {\mathbf{0}}&{{\mathbf{I}_k}} \\ 
  {{{\mathbf{I}}_{N - k}}}&{\mathbf{0}}        
    \end{bmatrix}.
\end{equation}

In both cases, the sidelobe level of the ACF plays a critical role.  Let us take the A-ACF as an example. The sidelobe of $r_k$ is defined as
    \begin{equation}
            |r_k|^2  = |\mathbf{x}^H \mathbf{J}_k \mathbf{x}|^2 = |r_\mathit{-k}|^2, \quad {k = 1,2,..,N-1} ,  
    \end{equation}
where $|r_0|^2$=$|\mathbf{x}^H\mathbf{x}|^2$= $|\mathbf{s}^H \mathbf{B}^H \mathbf{F}_N \mathbf{F}_N^H \mathbf{B} \mathbf{s}|^2$ = $\left| \sum_{i=1}^N P_i \left| s_i\right|^2  \right|^2$  is the mainlobe of the ACF. The integrated sidelobe level (ISL) can be expressed as \cite{Song2015}, \cite{Stoica2009TSP}
    \begin{equation}\label{ISL}
        \mathrm{ISL} = \sum_{k=1}^{N-1} \left|r_k\right|^2.    
    \end{equation}

Given the random nature of the signal $\mathbf{x}$, we define the average sidelobe level as the primary sensing performance metric. This metric is expressed as
    \begin{equation}
            \mathbb{E}(|r_k|^2)  = \mathbb{E} (|\mathbf{x}^H \mathbf{J}_k \mathbf{x}|^2 ) = \mathbb{E} ( |\mathbf{s}^H \mathbf{B}^H \mathbf{F}_N \mathbf{J}_k \mathbf{F}_N^H \mathbf{B} \mathbf{s}| ^2 ), \quad \forall k.
    \end{equation}
When $k$ is zero, $\mathbb{E}(|r_0|^2)$ denotes the average mainlobe level. Correspondingly, the EISL is given by
    \begin{equation}\label{EISL}
        \mathrm{EISL} = \sum_{k=1}^{N-1}\ \mathbb{E}(|r_k|^2) .    
    \end{equation}
Consequently, seeking the optimal PA strategy is equivalent to solving the following stochastic optimization problem:
   \begin{equation} \label{stochastic-optimization-problem-1}
       \mathop{\text{min}}\limits_{\sum_{i=1}^N P_i = N, P_i \geq 0}  \frac{\sum_{k=1}^{N-1} \mathbb{E} \left(\left |\mathbf{s}^H \mathbf{B}^H \mathbf{F}_N \mathbf{J}_k \mathbf{F}_N^H \mathbf{B} \mathbf{s} \right|^2  \right)} { \mathbb{E}\left(\left| \sum_{i=1}^N P_i \left| s_i\right|^2  \right|^2\right)  },   
   \end{equation}
which is the EISL normalized by the expected mainlobe level.

\section{Statistical Characterizations of P-ACF and A-ACF}\label{Sec_III}
In \cite{liu2024ofdmachieveslowestranging}, the closed-form expressions of P-ACF and A-ACF without PA have been theoretically analyzed. In this section, we systematically investigate the closed-form expressions of both P-ACF and A-ACF with PA in OFDM systems. 

\subsection{The P-ACF Case} \label{The P-ACF Case with power allocation}
In this section, we present the main result of the P-ACF case with PA, which is generally easier to tackle than its A-ACF counterpart. Let us first demonstrate the basic structure of the periodic shift matrix in Lemma 1.

    \begin{lemma}              
        The periodic shift matrix can be decomposed as
            \begin{equation} \label{J_k fenjie}
                \widetilde {\mathbf{J}}_k = \sqrt{N} \mathbf{F}_N^H \mathrm{Diag}(\mathbf{f}_\mathit{N-k+1}) \mathbf{F}_N ,
            \end{equation}
        where $\mathbf{f}_k$ is the $\mathnormal{k}$th column of $\mathbf{F}_N$.        
    \end{lemma}

The P-ACF with PA can be expressed as
    \begin{equation}\label{rk with power allocation-IDFT}
    \begin{split}
        \widetilde{r}_k   = \mathbf{x}^H \widetilde{\mathbf{J}}_k \mathbf{x} 
             & = \mathbf{c}^H \mathbf{F}_N \widetilde{\mathbf{J}}_k \mathbf{F}_N^H\mathbf{c},
    \end{split}        
    \end{equation}
where $k = 0,1,...,N-1$ and $\mathbf{c}$ is
      \begin{align}\label{c}
        \mathbf{c} = \mathbf{B}  \mathbf{s} = \left[s_1 \sqrt{P_1}, s_2 \sqrt{P_2}, ..., s_N \sqrt{P_N} \right]^T.
      \end{align} 
Using Lemma 1, \eqref{rk with power allocation-IDFT} can be recast as 
    \begin{equation}
        \begin{split}\label{rk with power allocation-IDFT 2}
            \widetilde{r}_k  &= \sqrt{N} \mathbf{c}^H  \mathrm{Diag} \left(\mathbf{f}_\mathit{N-k+1}\right) \mathbf{c} \\
            & = \sum_{n=1}^N \left|c_n\right|^2 e^\frac{-j2\pi (N-k)(n-1)}{N}\\
            &= \sum_{n=1}^N \left|c_n\right|^2 e^\frac{j2\pi k(n-1)}{N} .
        \end{split}
    \end{equation}

\begin{proposition}
    The average squared P-ACF with PA is
    \begin{equation}\label{EAF with power allocation-IDFT 1}
        \mathbb{E} \left( \left|\widetilde{r}_k \right|^2 \right) = \left(\mathbf{\mu}_4-1\right) \sum_{i=1}^N P_i^2 + \left| \sum_{n=1}^N P_n e^\frac{j2\pi kn}{N} \right|^2 .   
    \end{equation}

\end{proposition}
\begin{proof}
    See Appendix \ref{Proof of Proposition1}.
\end{proof}

As a consequence, the average mainlobe of the P-ACF with PA is
    \begin{equation}\label{EAF with power allocation-IDFT main}
        \ \mathbb{E}\left( \left|\widetilde{r}_0\right|^2\right) = (\mathbf{\mu}_4-1) \sum_{i=1}^N P_i^2 + N^2 .     
    \end{equation}

\begin{corollary}
The EISL of the P-ACF with PA is
    \begin{equation}\label{EAF with power allocation-IDFT sibe}
        \sum_{k=1}^{N-1}\ \mathbb{E} \left( \left| \widetilde{r}_k \right|^2 \right) = \left[(N-1)\mu_4+1\right] \sum_{i=1}^N P_i^2 - N^2 .      
    \end{equation}     
\end{corollary}

\begin{proof}
    See Appendix \ref{Proof of Proposition2}.
\end{proof}

\subsection{The A-ACF Case} \label{IV-B}
In this section, we evaluate the ranging performance of random ISAC signals with PA by investigating the sidelobe level of A-ACF.

Similarly to \eqref{rk with power allocation-IDFT}, the expression of A-ACF with PA is given by
    \begin{equation}\label{rk with power allocation-IDFT A-ACF}
    \begin{split}
        r_k   = \mathbf{x}^H \mathbf{J}_k \mathbf{x} 
              = \mathbf{s}^H \mathbf{B}^H \mathbf{F}_N \mathbf{J}_k \mathbf{F}_N^H\mathbf{B}\mathbf{s} ,
    \end{split}        
    \end{equation}
where $\mathbf{J}_k$ is a shift matrix and the size is $N \times N$. To apply Lemma 1, \eqref{rk with power allocation-IDFT A-ACF} can be rewritten as 
    \begin{equation} \label{rk with power allocation-IDFT A-ACF 2}
    \begin{split}
        {r}_k= \mathbf{x}^H \mathbf{J}_k \mathbf{x}  = 
        \begin{bmatrix}
            \mathbf{x}^H 
            \mathbf{0}_N^H
        \end{bmatrix} 
        \widetilde{\mathbf{J}}^{\star}_{k}
        \begin{bmatrix}
            \mathbf{x} \\
            \mathbf{0}_N
        \end{bmatrix}      ,  
    \end{split}
    \end{equation}
where ${\widetilde{\mathbf{J}}^{\star}_{k}}$ is a periodic shift matrix with size $2N \times 2N$ and $\mathbf{0}_{N}$ is a all-zero vector of length $N$. According to Lemma 1, it can be decomposed as
    \begin{equation} \label{J_k fenjie 2}
        \widetilde {\mathbf{J}}^{\star}_{k} = \sqrt{2N} \mathbf{F}_{2N}^H \mathrm{Diag}(\mathbf{f}^{\star}_{2N-k+1}) \mathbf{F}_{2N} ,
    \end{equation}
where $\mathbf{f}^{\star}_{k}$ is the $\mathnormal{k}$th column of $\mathbf{F}_{2N}$, which is the normalized DFT matrix of size $2N$. 

Accordingly, \eqref{rk with power allocation-IDFT A-ACF 2} can be rewritten as
    \begin{equation} \label{rk with power allocation-IDFT A-ACF 3}
    \begin{split}
        {r}_k =  
         \sqrt{2N} \mathbf{s}^H \mathbf{B}^H \mathbf{F}_N \widetilde{\mathbf{F}}_{2N}^H \mathrm{Diag}(\mathbf{f}^{\star}_{2N-k+1}) \widetilde{\mathbf{F}}_{2N} \mathbf{F}_N^H \mathbf{B} \mathbf{s} ,
    \end{split}
    \end{equation}
    where $\widetilde{\mathbf{F}}_{2N}\in \mathbb{C}^{2N \times N}$ contains the first $N$ columns of $\mathbf{F}_{2N}$. By denoting $\mathbf{V} = \mathbf{B}^H \mathbf{F}_N \widetilde{\mathbf{F}}_{2N}^H = [\mathbf{v}_1,\mathbf{v}_2,...,\mathbf{v}_{2N}]$, it follows that 
    \begin{align} \label{rk with power allocation-IDFT A-ACF 4}
        \ {r}_k  \nonumber& = \sqrt{2N} \mathbf{s}^H \mathbf{V} \mathrm{Diag}(\mathbf{f}^{\star}_{2N-k+1}) \mathbf{V}^H \mathbf{s}\\
        \nonumber& = \sum_{n=1}^{2N} \left| \mathbf{v}_n^H \mathbf{s} \right|^2 e^{-\frac{j 2\pi (2N-k)(n-1)}{2N}} \\
       & = \sum_{n=1}^{2N} \left| \mathbf{v}_n^H \mathbf{s} \right|^2 e^{\frac{j   2\pi k(n-1)}{2N}} .
    \end{align}
The average mainlobe level and the EISL of the A-ACF are $\mathbb{E}(|r_0|^2)$ and $\frac{1}{2} \sum_{k=1}^{2N-1}\mathbb{E}(|{r}_k|^2)$ respectively. 

\begin{proposition}
The average mainlobe and the EISL of the A-ACF with PA are respectively
\begin{align}
    \mathbb{E}(|r_0|^2) &= (\mu_4-1) \sum_{i=1}^{N} P_{i}^2 + N^2 \label{maiblobe of A-ACF power allocation} ,\\ \nonumber
    \frac{1}{2} \sum_{k=1}^{2N-1}\mathbb{E}(|{r}_k|^2) & = N(\mu_4 - 2) \|\mathbf{V}\|_{4}^{4} + 2N \sum_{n=1}^{2N} \|\mathbf{v}_n\|_{2}^{4} \label{EISL of A-ACF power allocation} \\ 
        & - \frac{1}{2} \left[ (\mu_4-1) \sum_{i=1}^{N} P_{i}^2 + N^2 \right] . 
\end{align}
\end{proposition}    

\begin{proof}
    See Appendix \ref{Proof of proposition 3}.
\end{proof}

\section{Optimal Power Allocation schemes}\label{Sec_IV}
In this section, we analyze the optimal PA based on OFDM modulation for both the P-ACF and A-ACF cases.

\subsection{The P-ACF Case}

\begin{thm} [Global Optimality of the Uniform PA for Ranging]
    When the OFDM signaling scheme is employed, the uniform PA is the only scheme that achieves the lowest normalized EISL of the P-ACF across all constellations.
\end{thm}

\begin{proof}
It is evident that both the average mainlobe and the EISL vary with the PA scheme. Therefore, we normalize the EISL with respect to the average mainlobe, yielding the expression  
    \begin{equation}\label{normalize EAF with power allocation-IDFT sibe}
    \begin{split}
        \frac{\sum_{k=1}^{N-1}\ \mathbb{E}(\left| \widetilde{r}_k \right|^2)}{\mathbb{E}(\left| \widetilde{r}_0 \right|^2)}
        &= \frac{N\mathbf{\mu}_4 \sum_{i=1}^N P_i^2}{(\mathbf{\mu}_4-1) \sum_{i=1}^N P_i^2 + N^2} - 1  \\
        &= \frac{N\mathbf{\mu}_4}{(\mathbf{\mu}_4-1) + \frac{N^2}{\sum_{i=1}^N P_i^2}} - 1 ,
    \end{split}
    \end{equation}
indicating that the normalized EISL is minimized if and only if $\sum_{i=1}^N P_i^2$ is minimized. Moreover, note that
    \begin{equation} \label{value of P}
        \begin{split}
            \sum_{i=1}^N P_i^2 = \frac{1}{N} \sum_{i=1}^N P_i^2 \sum_{i=1}^N 1^2 \geq \frac{1}{N} \left( \sum_{i=1}^N P_i\right)^2
        \end{split}
    \end{equation}
due to the Cauchy-Schwarz Inequality. The equality is attained if and only if $P_1 = P_2 = ... = P_N = 1$. This suggests that the normalized EISL is minimized under uniform PA, completing the proof.  
\end{proof}

\begin{thm} [The Uniform PA achieves the Lowest Sidelobe at Every Lag] When the OFDM signaling scheme is used, the uniform PA is the only scheme that achieves the lowest sidelobe level at every delay index $k$ of the P-ACF for all constellations.    
\end{thm}
\begin{proof}
    According to \eqref{EAF with power allocation-IDFT 1} and \eqref{EAF with power allocation-IDFT main}, the normalized $\mathbb{E}(|\widetilde{r}_k|^2)$ is shown as follows:
    \begin{equation} \label{normalized EAF with power allocation-IDFT 1}
    \begin{aligned}
        \frac{\mathbb{E}(|\widetilde{r}_k|^2)}{\mathbb{E}(\left| \widetilde{r}_0 \right|^2)} 
        & = 1 - \frac{N^2 - \left| \sum_{n=1}^N P_n e^\frac{j2\pi kn}{N} \right|^2}{(\mu_4-1) \sum_{i=1}^N P_i^2 + N^2} , k = 1,2,...,N-1.
    \end{aligned}
    \end{equation}
When $P_1 = P_2 = ... = P_N = 1$, 
\begin{equation}
    \left| \sum_{n=1}^N P_n e^\frac{j2\pi kn}{N} \right|^2 = 0,
\end{equation}
which maximizes the numerator of the second term of the right-hand side (RHS) in \eqref{normalized EAF with power allocation-IDFT 1}. Moreover, by applying an argument similar to the proof of Theorem 1, the denominator is minimized, which, in turn, maximizes the second term of the RHS of \eqref{normalized EAF with power allocation-IDFT 1}. Therefore, only when $P_1 = P_2 = ... = P_N = 1$, the normalized $\mathbb{E}(|\widetilde{r}_k|^2)$ reaches to its minimum, completing the proof.
\end{proof}

When the uniform PA scheme is employed, the average mainlobe level is
    \begin{equation}\label{EAF with uniform power allocation-IDFT main }
    \begin{split}
        \ \mathbb{E}(|\widetilde{r}_0|^2) = ({\mu_4}-1)N + N^2 ,
    \end{split}        
    \end{equation}
the average sidelobe level is
    \begin{equation}\label{EAF with uniform power allocation-IDFT}
    \begin{split}
        \mathbb{E}(|\widetilde{r}_k|^2) = ({\mu_4}-1)N,\quad k \ne 0 ,
    \end{split}        
    \end{equation}
and the normalized EISL is expressed as 
    \begin{equation}\label{normalize EAF with uniform power allocation-IDFT sibe}
    \begin{split}
        \frac{\sum_{k=1}^{N-1}\ \mathbb{E}(|\widetilde{r}_k|^2)}{\mathbb{E}(|\widetilde{r}_0|^2)}
        = \frac{N{\mu_4}}{({\mu_4}-1) + N} - 1 .
    \end{split}
    \end{equation}

An interesting observation is that for PSK constellations ($\mu_4 = 1$), both the individual and the integrated sidelobe levels are zero under OFDM signaling with uniform power. This implies that the P-ACF of an OFDM-PSK signal with uniform power is always a unit impulse function. This observation is consistent with the results introduced in \cite{liu2024ofdmachieveslowestranging}.

\subsection{The A-ACF Case}
\begin{thm} [Uniform PA is Globally Optimal for All Constellations] When the OFDM signaling scheme is used, the uniform PA is the only PA scheme that achieves the lowest normalized EISL of the A-ACF for all constellations.
\end{thm}
\begin{proof}
    See Appendix \ref{Proof of theorem 3}
\end{proof}

At this time, the average mainlobe of A-ACF is
    \begin{equation} \label{maiblobe of A-ACF avarege power}
        \mathbb{E}(|r_0|^2) = (\mu_4-1) N + N^2,
    \end{equation}
and the normalized EISL of A-ACF is
    \begin{equation} \label{normalized EISL of A-ACF with uniform power}
    \begin{split}
        \frac{\frac{1}{2} \sum_{k=1}^{2N-1} \mathbb{E}(|r_k|^2)}{\mathbb{E}(|r_0|^2)} = \frac{N(\mu_4 - 2) \left| \left| \mathbf{F}_N \widetilde{\mathbf{F}}_{2N}^{H} \right| \right|_{4}^{4} + N^2}{(\mu_4-1) N + N^2} - \frac{1}{2}.
    \end{split}    
    \end{equation}

\section{CP-OFDM With power allocation and frequency zero-padding} \label{Sec V}
In this section, we analyze the P-ACF for the OFDM scheme incorporating PA and frequency zero-padding. Zero-padding refers to the process of appending zeros to the end of a signal or data sequence to extend its length. Specifically, zero-padding in the frequency domain involves adding zero values to the high-frequency portion of the signal's spectrum, thereby increasing the overall length of the spectrum. By doing this, the resolution of the time-domain signal is enhanced through increasing the number of sampling points. Consequently, the accuracy of the range estimation may be improved.

 According to \eqref{rk with power allocation-IDFT 2}, when the OFDM modulation and PA are applied, the P-ACF is the inverse discrete Fourier
transform (IDFT) of the sequence $|c_n|^2$ up to a factor $\frac{1}{N}$. Through frequency zero-padding with $N(L-1)$ null subcarriers, where $L\ge1$ represents the padding factor, we generate an augmented sequence $\hat{\mathbf{c}}$ with an extended length of $NL$ samples, yielding
    \begin{equation}
    \begin{split}
        \hat{\mathbf{c}} & = \left[|c_1|^2,|c_2|^2,...,|c_N|^2,\mathbf{0}_{N(L-1)}{^T} \right]^T \\
        & = \left[|s_1\sqrt{P_1}|^2,|s_2\sqrt{P_2}|^2,...,|s_N\sqrt{P_N}|^2,\mathbf{0}_{N(L-1)}{^T} \right]^T  ,    
    \end{split}
    \end{equation}
where $\mathbf{0}_{N(L-1)}$ represents the all-zero vector with length $N(L-1)$.

Analogous to the derivation in \eqref{rk with power allocation-IDFT 2}, we construct an enhanced P-ACF incorporating both PA and frequency-domain zero-padding. The modified P-ACF, denoted as $\hat{r}_k$, is expressed as
    \begin{equation}
        \begin{split}\label{rk with power allocation and zero padding}
            \hat{r}_k &= \sum_{n=1}^N \left| c_n \right|^2 e^\frac{j2\pi k(n-1)}{NL} 
            = \sum_{n=1}^N {P_n} \left| s_n \right|^2 e^\frac{j2\pi k(n-1)}{NL} ,\\
        \end{split}
    \end{equation}
which is the IDFT of the sequence $\hat{\mathbf{c}}$ up to a factor $\frac{1}{NL}$.

\begin{proposition}
    The expectation of the integrated sidelobe level EISL$_{ZP}$ and the average mainlobe level of the P-ACF with PA and frequency zero-padding can be expressed as 
    \begin{align}\nonumber
        \operatorname{EISL}_{\text{ZP}} &= \sum_{k=L}^{\frac{NL}{2}-1} \mathbb{E} (|\hat{r}_k|^2)  \\
        & = (\frac{NL}{2}-L) \left( \mu_4-1 \right) \left\| \mathbf{p} \right\|_2^2 + \left\| \mathbf{G}^H \mathbf{p} \right\|_2^2 \label{hat EISL}, \\ 
        \mathbb{E} (|\hat{r}_0|^2) &= (\mu_4-1) \sum_{i=1}^N P_i^2 + N^2,
    \end{align}
where
    \begin{align} 
        \mathbf{g}_k &= 
        \left[
            e^\frac{-j2\pi k}{NL},
            e^\frac{-j2\pi 2k}{NL},
            . 
            . 
            . 
            ,e^\frac{-j2\pi Nk}{NL} \right]^T ,\\
        \mathbf{G} &= 
        \left[
            \mathbf{g}_L, \mathbf{g}_{L+1}, ... ,\mathbf{g}_{NL/2-1}
        \right],\\
        \mathbf{p} &= \left[P_1,P_2,...,P_N\right]^T.
    \end{align}
\end{proposition}
    
\begin{proof}
    See Appendix \ref{proof of proposition 4}.
\end{proof}

 The normalized $\text{EISL}_{\text{ZP}}$ can be denoted as 
    \begin{equation} \label{normalized hat EISL}
            f(\mathbf{p}) =\frac{\left( \frac{NL}{2}-L \right) \left( \mu_4-1 \right) \left\| \mathbf{p} \right\|_2^2 + \left\| \mathbf{G}^H \mathbf{p} \right\|_2^2}{(\mu_4-1) \left\| \mathbf{p} \right\|_2^2 + N^2} .
    \end{equation}


Through comprehensive theoretical analysis and numerical simulations, we found that the uniform PA scheme ceases to be optimal for minimizing the normalized $\text{EISL}_{\text{ZP}}$ when frequency zero-padding is implemented. To address this limitation, we formulate the following constrained optimization problem to derive a suitable PA strategy that can minimize the normalized $\text{EISL}_{\text{ZP}}$.

   \begin{equation} \label{objective function power allocation zero-padding}
   \begin{aligned}
       &\min  \quad f(\mathbf{p}) \\
       &\mathrm{s.t.} \quad  \sum\nolimits_{i=1}^{N} P_i = N  \\  
       & \qquad \quad P_i \ge 0 , i = 1,2,...,N. \\
   \end{aligned}
   \end{equation}


Notably, the optimization problem in \eqref{objective function power allocation zero-padding} requires minimizing a fractional objective function over the simplex. To solve this constrained problem, we adopt the PGD algorithm to obtain a locally optimal solution. As outlined in Algorithm 1, the PGD method enhances the conventional Gradient Descent (GD) by introducing a projection step, which ensures that each iteration remains within the feasible region \cite{boyd2004convex}, \cite{Lushihang2024TSP}. The gradient of the objective function  $f(\mathbf{p})$ at an arbitrary point $\mathbf{p}$ is given by
    \begin{equation} \label{the gradient of p power allocation zero-padding}
        \nabla f(\mathbf{p}) = \frac{ \frac{\partial \text{EISL}_\text{ZP}}{\partial \mathbf{p}} \mathbb{E}(|\hat{r}_0|^2) - \text{EISL}_\text{ZP} \frac{\partial \mathbb{E}(|\hat{r}_0|^2)}{\partial \mathbf{p}}} {(\mathbb{E}(|\hat{r}_0|^2))^2},
    \end{equation}
where
    \begin{align}\nonumber
        \frac{\partial \text{EISL}_\text{ZP}}{\partial \mathbf{p}} &= \left( NL - 2L \right) \left( \mu_4 -1 \right) \mathbf{p} \\
        &+ \sum_{k=L}^{\frac{NL}{2} - 1} \left( \mathbf{g}_k^* \mathbf{p}^H \mathbf{g}_k + \mathbf{g}_k \mathbf{p}^T \mathbf{g}_k^*\right) ,\\
        \frac{\partial \mathbb{E}(|\hat{r}_0|^2)}{\partial \mathbf{p}} &= 2(\mu_4 - 1) \mathbf{p}.
    \end{align}

\begin{center}
\begin{algorithm}[!t]
    \caption{PGD Algorithm for Solving \eqref{objective function power allocation zero-padding}}
    \label{alg:PGD Project 2}
    \begin{algorithmic}[1]
    \REQUIRE $N$, $\mathbf{p}_0$, $r_{max}$, $\varepsilon$
    \ENSURE $\mathbf{p}$
    \STATE Initialize $r = 1$ and $\mathbf{p}^{(1)} = \mathbf{p}_0$.
    \STATE Initialize the Communication Symbol Set $\mathbf{S} = \left\{ \mathbf{s}_1,\mathbf{s}_2,...,\mathbf{s}_N \right\}$ and calculate $\mu_4$.
    \STATE $\textbf{repeat}$
    \STATE \quad Calculate ${\nabla}f(\mathbf{p}^{(r)})$.
    \STATE \quad Update $\mathbf{p}^{(r+1)}$ $\xleftarrow{}$ 
    $ \mathrm{Proj}_\mathcal{B}(\mathbf{p}^{(r)} - \eta^{(r)}{\nabla}f(\mathbf{p}^{(r)}) )$.
    \STATE \quad Update $r = r+1$.
    \STATE $\textbf{until}$ the decrease of the objective value is below $\varepsilon$ or $r = r_{max}$.    
\end{algorithmic}
\end{algorithm}
\end{center}
Moreover, the projector over the simplex can be expressed as
    \begin{equation} \label{the projection operation p}
        \quad \mathrm{Proj}_\mathcal{B} (\mathbf{p}^{(r+1)}) =  \frac{N}{\mathbf{1}_N^T \mathbf{p}^{\left(r+1\right)} }  \mathbf{p}^{\left(r+1\right)},
    \end{equation}
where $\mathbf{1}_N$ represents the all-one vector with length $N$ and $\mathcal{B}$ denotes the feasible region of the optimization variable $\mathbf{p}$.

The average sidelobe levels achieved by the proposed PA scheme, obtained via Algorithm \ref{alg:PGD Project 2}, are illustrated in Fig. \ref{fig:P-ACF 16-QAM power different N} and Fig. \ref{fig:P-ACF 16-PSK power different N}. In these figures, ``optimal PA'' denotes the PA scheme derived from Algorithm \ref{alg:PGD Project 2}. The results indicate that the optimal PA offers considerable sidelobe suppression compared to the ``uniform PA''. However, this improvement is accompanied by a slight broadening of the mainlobe, which in turn leads to a degradation in range resolution performance.


To mitigate the adverse effects of PA on mainlobe width, we propose to add a new constraint in the original optimization problem \eqref{objective function power allocation zero-padding}. The enhanced formulation is given by:

   \begin{equation} \label{objective function power allocation zero-padding with width}
   \begin{aligned}
       &\min  \quad f(\mathbf{p}) \\
       &\mathrm{s.t.} \quad  \sum\nolimits_{i=1}^{N} P_i = N  \\  
       & \qquad \quad P_i \ge 0 , i = 1,2,...,N \\
       & \qquad \quad \mathbb{E} (|\hat{r}_q|^2)  \le \frac{1}{2} \mathbb{E} (|\hat{r}_0|^2).
   \end{aligned}
   \end{equation}
The newly added constraint $\mathbb{E} (|\hat{r}_q|^2) \le \frac{1}{2} \mathbb{E} (|\hat{r}_0|^2)$ ensures that the average sidelobe level at delay 
$q$ remains within 3 dB of the average mainlobe level $\mathbb{E} (|\hat{r}_0|^2)$. This condition is equivalent to $ 10 \text{log}(\frac{\mathbb{E} (|\hat{r}_q|^2)}{\mathbb{E} (|\hat{r}_0|^2)}) \le -3 \text{dB}$.


\begin{center}
\begin{algorithm}[!t]
    \caption{SCA Algorithm for Solving \eqref{objective function power allocation zero-padding with width}}
    \label{alg:SCA algorithm}
    \begin{algorithmic}[1]
    \REQUIRE $N$, $L$, $r_{max}$, $\varepsilon$
    \ENSURE $\mathbf{p}$
    \STATE Find a feasible starting point $\mathbf{p}^{(0)}$.
    \STATE Initialize the Communication Symbol Set $\mathbf{S} = \left\{ \mathbf{s}_1,\mathbf{s}_2,...,\mathbf{s}_N \right\}$ and calculate $\mu_4$.
    \STATE Initial $r = 0$.
    \STATE $\textbf{repeat}$
    \STATE \quad Calculate ${\nabla}f(\mathbf{p}^{(r)})$ and construct $f_{su}(\mathbf{p},\mathbf{p}^{(r)})$.
    \STATE \quad Attaining $\hat{\mathbf{p}}^{(r)}$ by solving the convex problem \eqref{SCA surrogate function}. 
    \STATE \quad Update $\mathbf{p}^{(r+1)} = \mathbf{p}^{(r)} + \alpha^{(r)} (\hat{\mathbf{p}}^{(r)} - {\mathbf{p}}^{(r)})$.
    \STATE \quad Update $r = r+1$.
    \STATE $\textbf{until}$ the decrease of the objective value is below $\varepsilon$ or $r = r_{max}$.    
\end{algorithmic}
\end{algorithm}
\end{center}

Compared to the original formulation in \eqref{objective function power allocation zero-padding}, the introduction of the third constraint in \eqref{objective function power allocation zero-padding with width} makes it challenging to project 
$\mathbf{p}$ onto the feasible region. As a result, the PGD algorithm is no longer directly applicable for solving the modified optimization problem. Upon analyzing the problem in \eqref{objective function power allocation zero-padding with width}, we observe that although all three constraints are convex, the objective function remains non-convex. Given these conditions, the Successive Convex Approximation (SCA) algorithm becomes a suitable alternative. SCA is a powerful framework for handling non-convex optimization problems, as it addresses non-convex optimization problems through iterative minimization of convex surrogate functions \cite{boyd2004convex, LiuAn2019TSP, KM2017}. The core procedure of the SCA algorithm lies in constructing an appropriate surrogate function $f_{su}(\mathbf{p},\mathbf{p}^{(r)})$ and determining a suitable step size 
$\alpha^{(r)}$. In Algorithm \ref{alg:SCA algorithm}, the surrogate function is defined as 
    \begin{equation} \label{SCA surrogate function}
    f_{su}(\mathbf{p}, \mathbf{p}^{(r)}) = f(\mathbf{p}^{(r)}) + (\nabla f(\mathbf{p}^{(r)}) )^H (\mathbf{p} - \mathbf{p}^{(r)}),
    \end{equation} 
which is the first-order Taylor expansion of the objective function near the $r$th iteration point. Furthermore, the step size 
$\alpha^{(r)}$ at each iteration is determined via the exact line search method.


\section{Simulation Results}\label{Sec_VI}

In this section, we further validate our conclusions and theoretical analysis by presenting numerical results. All the simulation results are attained by averaging over $1000$ random realizations. 

First, we focus on the P-ACF with PA of 16-QAM and 64-QAM. In Fig. \ref{fig: P-ACF OFDM uniform power simulation theoretical compare}, the sidelobe performances of 16-QAM and 64-QAM constellations for CP-OFDM with uniform PA are shown. It can be observed that the theoretical values align closely with the simulation results, confirming the accuracy of the derivations in Section \ref{The P-ACF Case with power allocation}. Additionally, a 3 dB performance gain in the peak-to-sidelobe level ratio is achieved by doubling $N$. In Fig. \ref{fig: P-ACF OFDM different power compare 64QAM}, the sidelobe performances of the P-ACF for 64-QAM with two PA schemes, uniform PA and random PA, are compared. It is clear that the uniform PA scheme achieves a much lower sidelobe level than the random PA scheme.

\begin{figure}[!t]
	\centering
	\includegraphics[scale = 0.425]{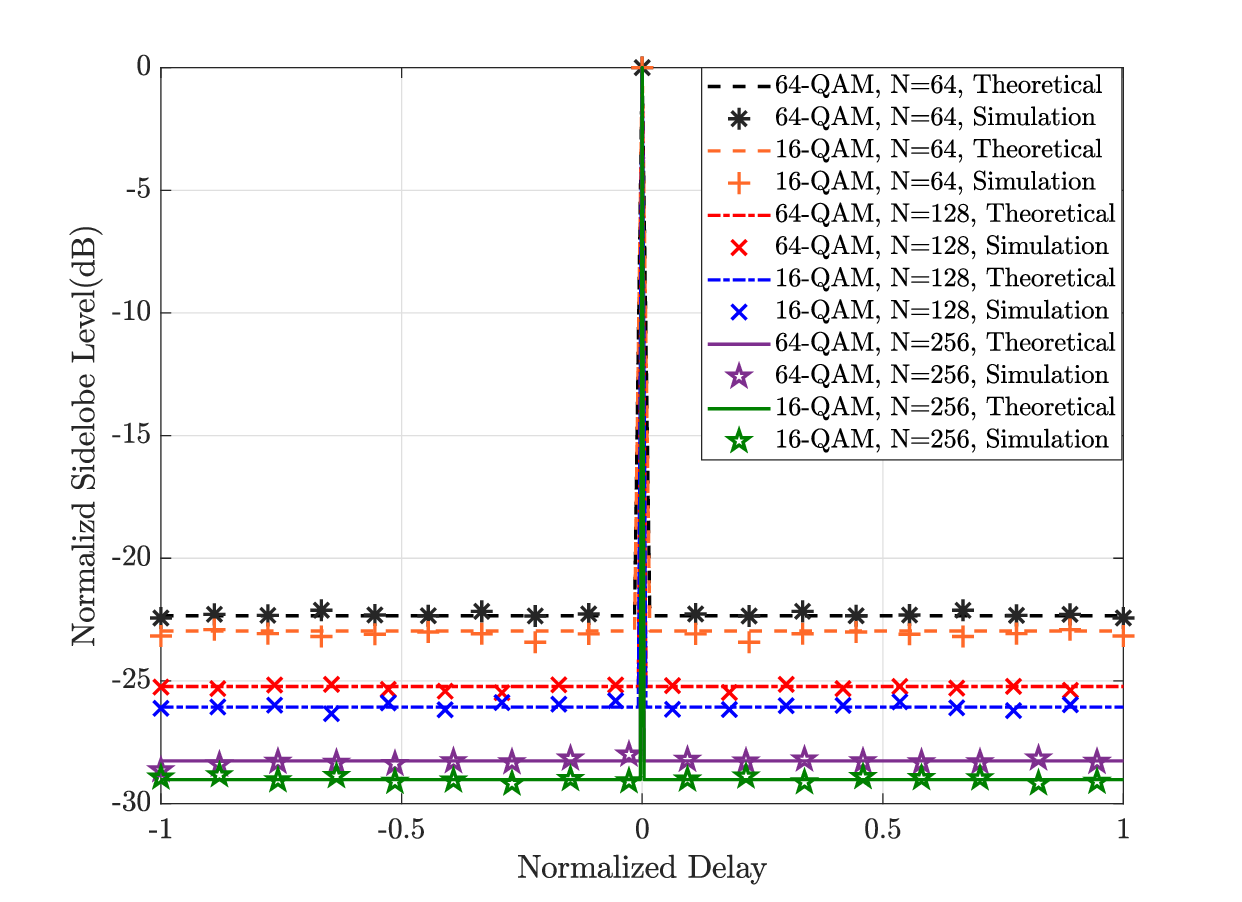} \vspace{-7mm}
	\caption{The P-ACF of 16/64-QAM under OFDM signaling and uniform power with varying $N$.}
    \label{fig: P-ACF OFDM uniform power simulation theoretical compare}
\end{figure} 
\begin{figure}[!t]
	\centering
	\includegraphics[width = \columnwidth]{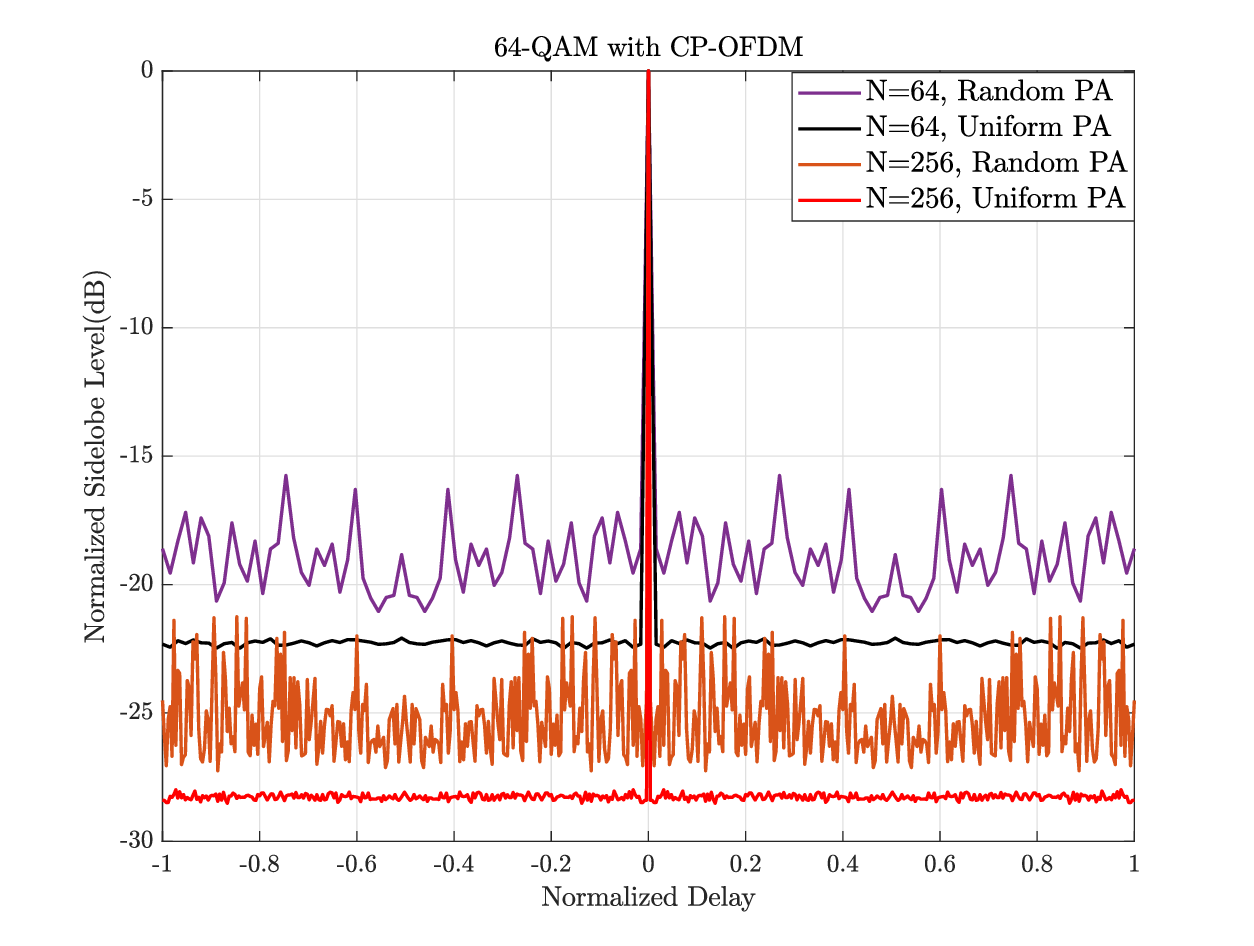} \vspace{-8mm}
	\caption{The P-ACF of 64-QAM under OFDM signaling and uniform power or random power, $N = 64/256$.}
    \label{fig: P-ACF OFDM different power compare 64QAM}
\end{figure}

In addition to the P-ACF, the sidelobe level performance of the A-ACF with PA is shown in Fig. \ref{fig:normalized EISL of A-ACF with power allocation} for the 16-QAM and 16-PSK constellations. We examine the normalized EISL without CP under different PA schemes. We observe that the theoretical results match perfectly with the numerical simulations for both the average and random PA schemes, confirming the accuracy of the theoretical derivations in Section \ref{IV-B}. Additionally, the normalized EISL for the uniform PA scheme is significantly lower than that of the random PA scheme in both the 16-QAM and 16-PSK constellations.

\begin{figure}[!t]
    \centering
    \includegraphics[scale = 0.415]{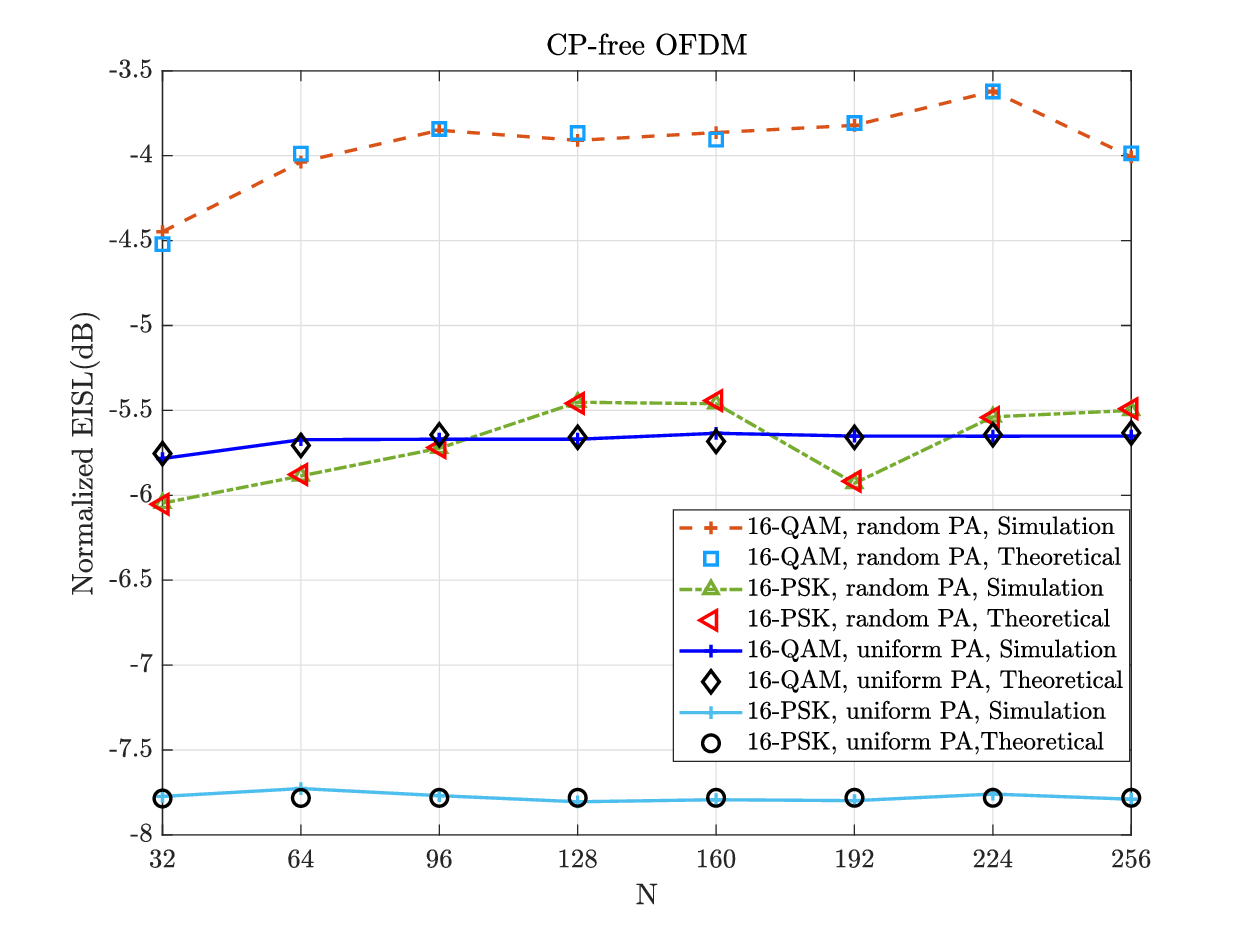}
    \vspace{-8mm}
    \caption{The resultant normalized EISL of A-ACF for 16PSK/16QAM under OFDM signaling with different power allocation and varying number of symbols.}
    \label{fig:normalized EISL of A-ACF with power allocation}
\end{figure}

\begin{figure}[!t]
	\centering
	\includegraphics[scale = 0.425]{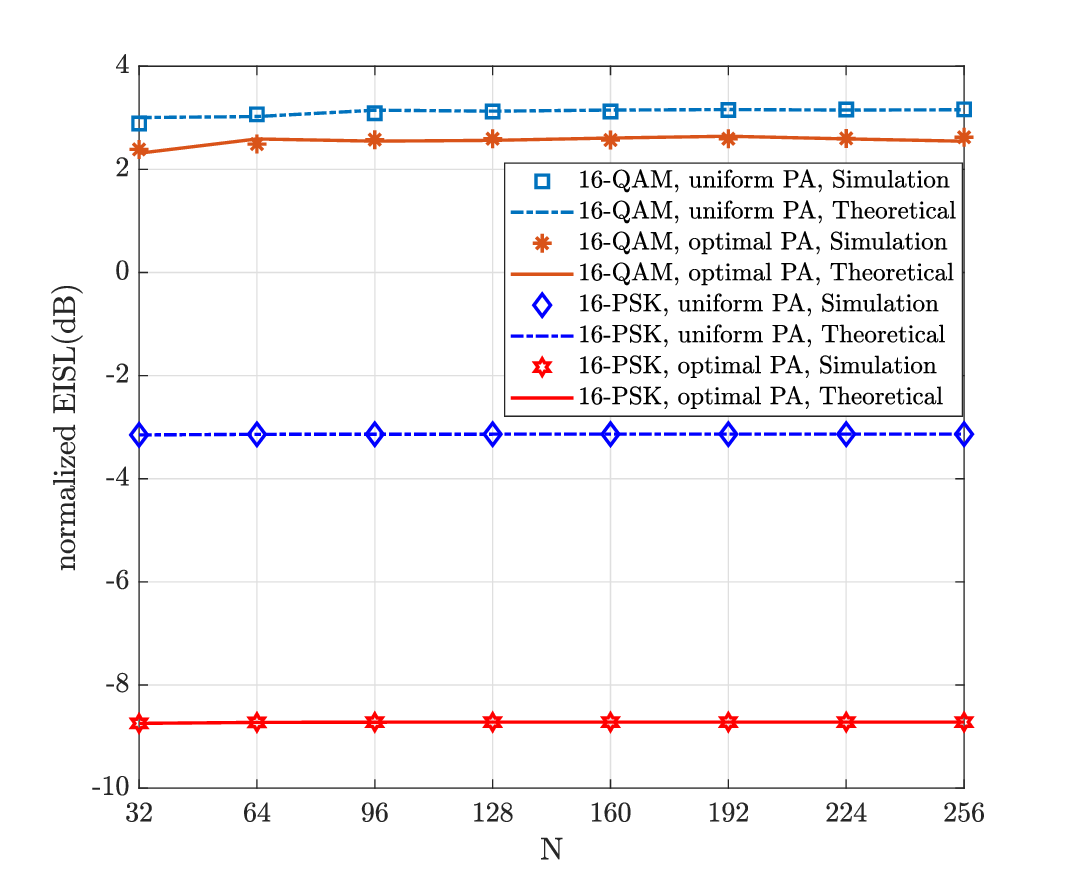}
 \vspace{-4mm}
	\caption{The resultant normalized $\text{EISL}_{\text{ZP}}$ for 16PSK/16QAM and $L=10$ under OFDM signaling with different power allocation and varying number of symbols.}
    \label{fig:P-ACF zero padding power allocation EISL compare}
\end{figure}

To verify the theoretical formulation of the normalized $\text{EISL}_{\text{ZP}}$ in \eqref{normalized hat EISL}, we present comparative simulations of computed 
$\text{EISL}_{\text{ZP}}$ values for different constellation schemes and PA methods in Fig. \ref{fig:P-ACF zero padding power allocation EISL compare}, where the proposed ``optimal PA'' scheme utilizes the power vector $\mathbf{p}$ obtained through Algorithm \ref{alg:PGD Project 2}. The close agreement between theoretical predictions and simulation measurements validates our analytical model. Notably, under the ``optimal PA'' scheme, the normalized $\text{EISL}_{\text{ZP}}$ for 16-PSK is significantly lower than that of the ``uniform PA'' case. For 16-QAM, the reduction is more modest, but the ``optimal PA'' still achieves a lower normalized $\text{EISL}_{\text{ZP}}$ compared to the uniform counterpart.


Next, we examine the sidelobe performances of 16-QAM and 16-PSK constellations with frequency zero-padding and different PA schemes. Fig. \ref{fig:P-ACF 16-QAM power different N} and Fig. \ref{fig:P-ACF 16-PSK power different N} display the P-ACF of 16-QAM and 16-PSK under OFDM signaling across different PA schemes and values of $N$. For both 16-QAM and 16-PSK, the ``optimal PA'' configuration achieves a lower sidelobe level than that of the ``uniform PA''. Analyzing these figures further, while the ``optimal PA'' reduces the sidelobe level, it also slightly increases the mainlobe width compared to the ``uniform PA'' scheme.

\begin{figure}[!t]
	\centering
	\includegraphics[width = \columnwidth]{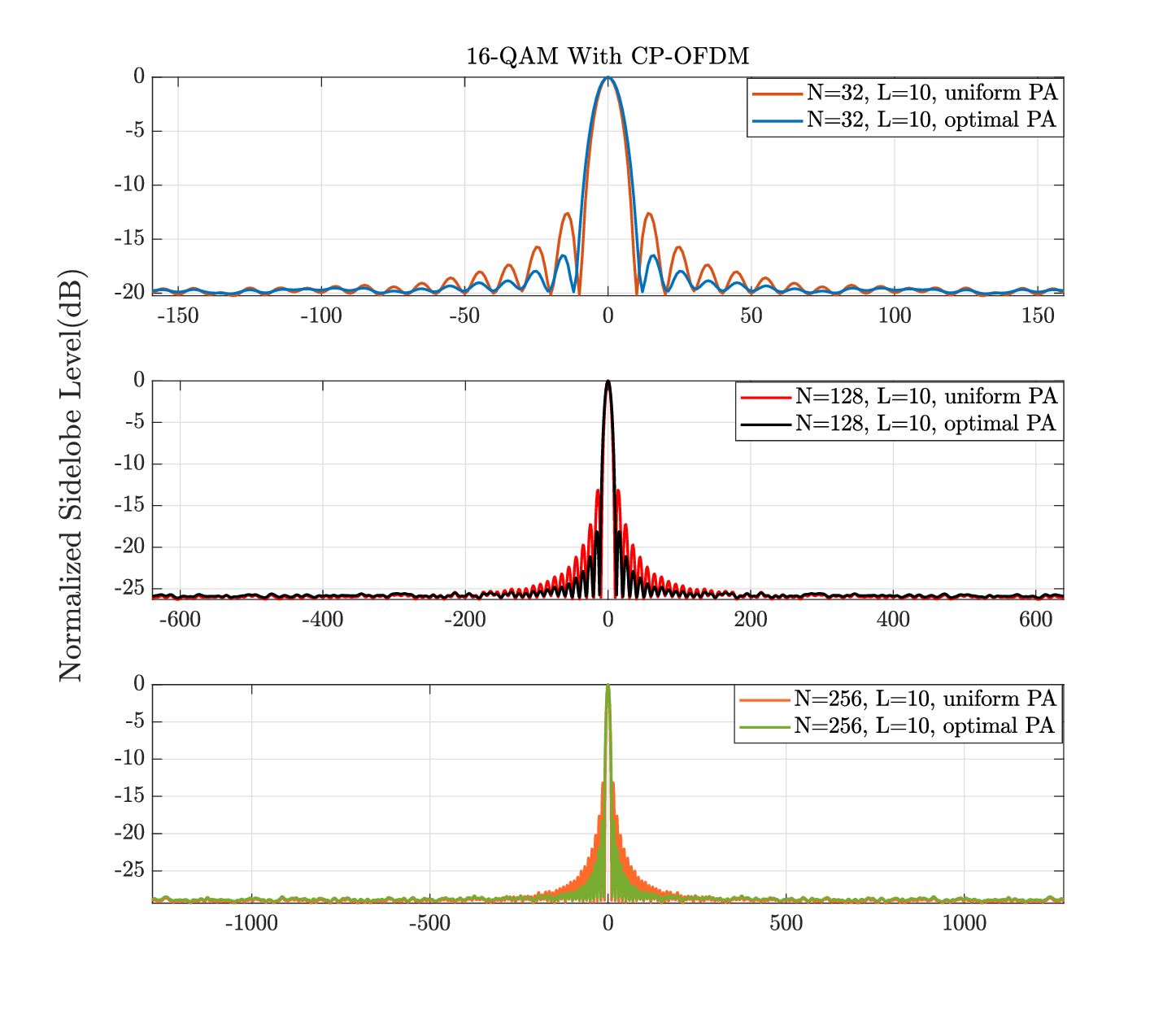}\vspace{-6mm}
	\caption{The P-ACF of 16-QAM under OFDM signaling and different power allocation schemes with varying $N$ and $L=10$.}
    \label{fig:P-ACF 16-QAM power different N}
\end{figure}

\begin{figure}[!t]
	\centering
	\includegraphics[width = \columnwidth]{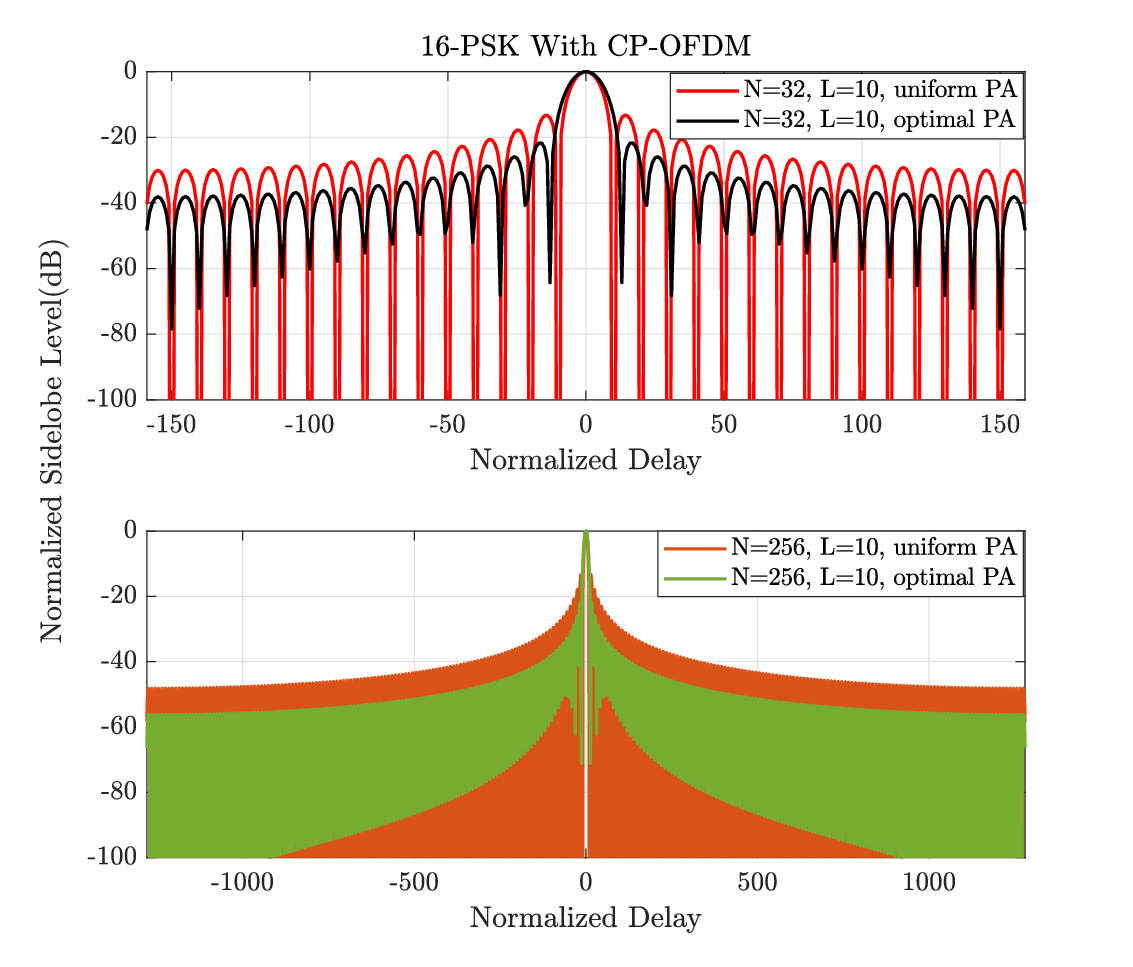}\vspace{-4mm}
	\caption{The P-ACF of 16-PSK under OFDM signaling and different power allocation schemes with varying $N$ and $L=10$.}
    \label{fig:P-ACF 16-PSK power different N}
\end{figure} 

In Fig. \ref{fig: oversampling P compare limit width}, we compare the convergence results of Algorithm \ref{alg:SCA algorithm} for 16-QAM and 16-PSK under different 3 dB mainlobe width constraints, where the limiting positions 
$q$ are set to 
$\frac{L}{2}$, $\frac{L}{2}+2$ and $\frac{L}{2}+4$, respectively. The ``uniform PA'' scheme and the ``unconstrained mainlobe width'' scheme serve as baseline references, where the latter refers to the PA vector obtained from Algorithm \ref{alg:PGD Project 2} without any constraint on mainlobe width. It is observed that as the 3 dB constraint becomes more stringent (i.e., smaller $q$), the convergence result increasingly resembles that of the ``uniform PA'' scheme. Conversely, when the constraint is relaxed (i.e., larger 
$q$), the convergence result tends to align more closely with the ``unconstrained mainlobe width'' case.

\begin{figure}[!t]
    \centering
    \begin{subfigure}
        \centering
        {\includegraphics[scale=0.3]{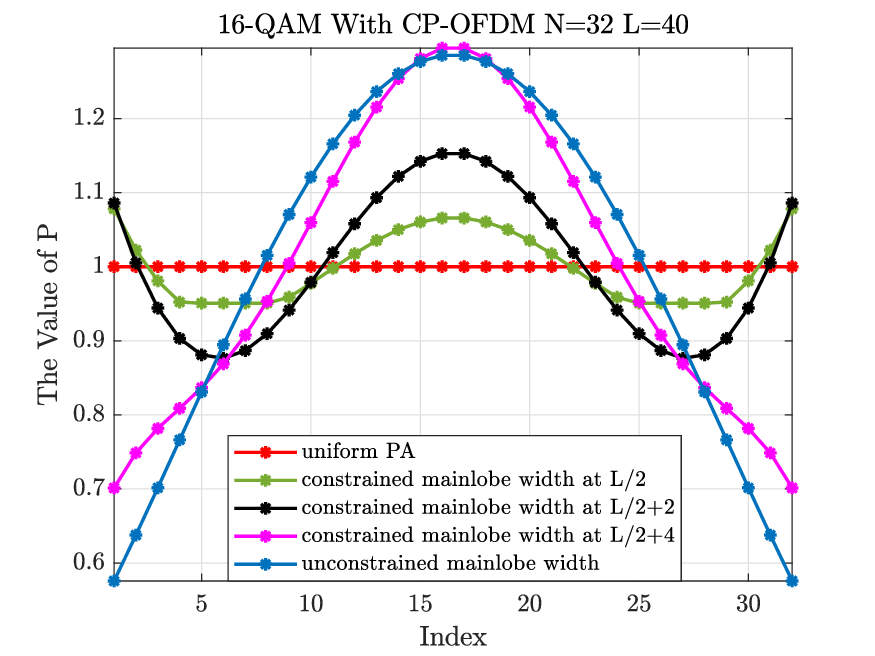}
        \label{fig: sub_figure1 of fig4}}
    \end{subfigure} \hspace{-6mm}
    \centering
    \begin{subfigure}
        \centering
        {\includegraphics[scale=0.3]{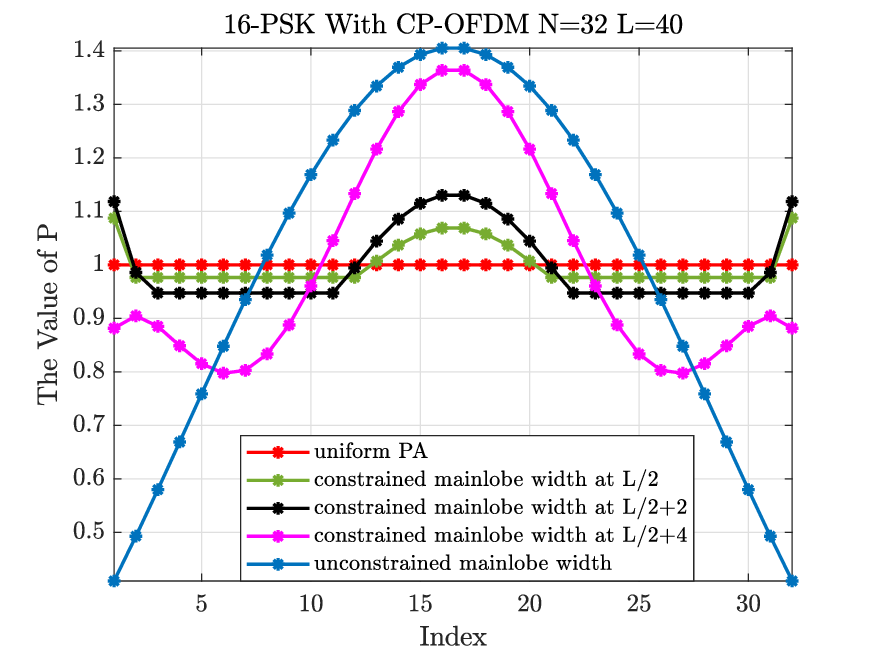}
        \label{fig: sub_figure2 of fig4}}
    \end{subfigure}    
    \centering
    \vspace{-8mm}
    \caption{The convergence results of the SCA algorithm for solving \eqref{objective function power allocation zero-padding with width} in 16-QAM and 16-PSK cases.}
    \label{fig: oversampling P compare limit width}
\end{figure}

\begin{figure}[!t]
	\centering
	\includegraphics[scale=0.4]{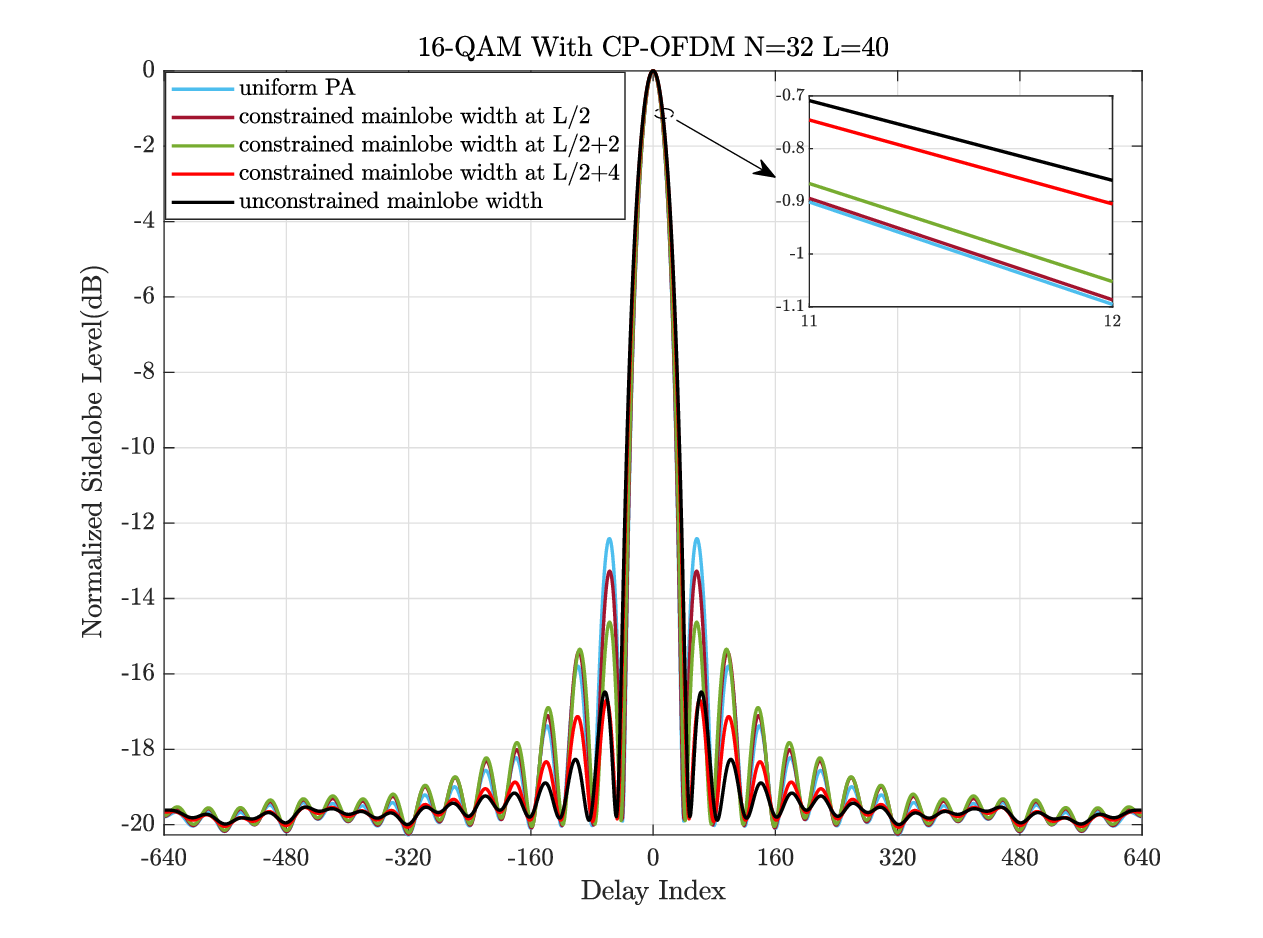} \vspace{-5mm}
	\caption{The P-ACF of 16-QAM under OFDM signaling and different power allocation schemes.}
    \label{fig: oversampling P-ACF compare limit width QAM}
\end{figure} 

\begin{figure}[!t]
	\centering
	\includegraphics[scale=0.42]{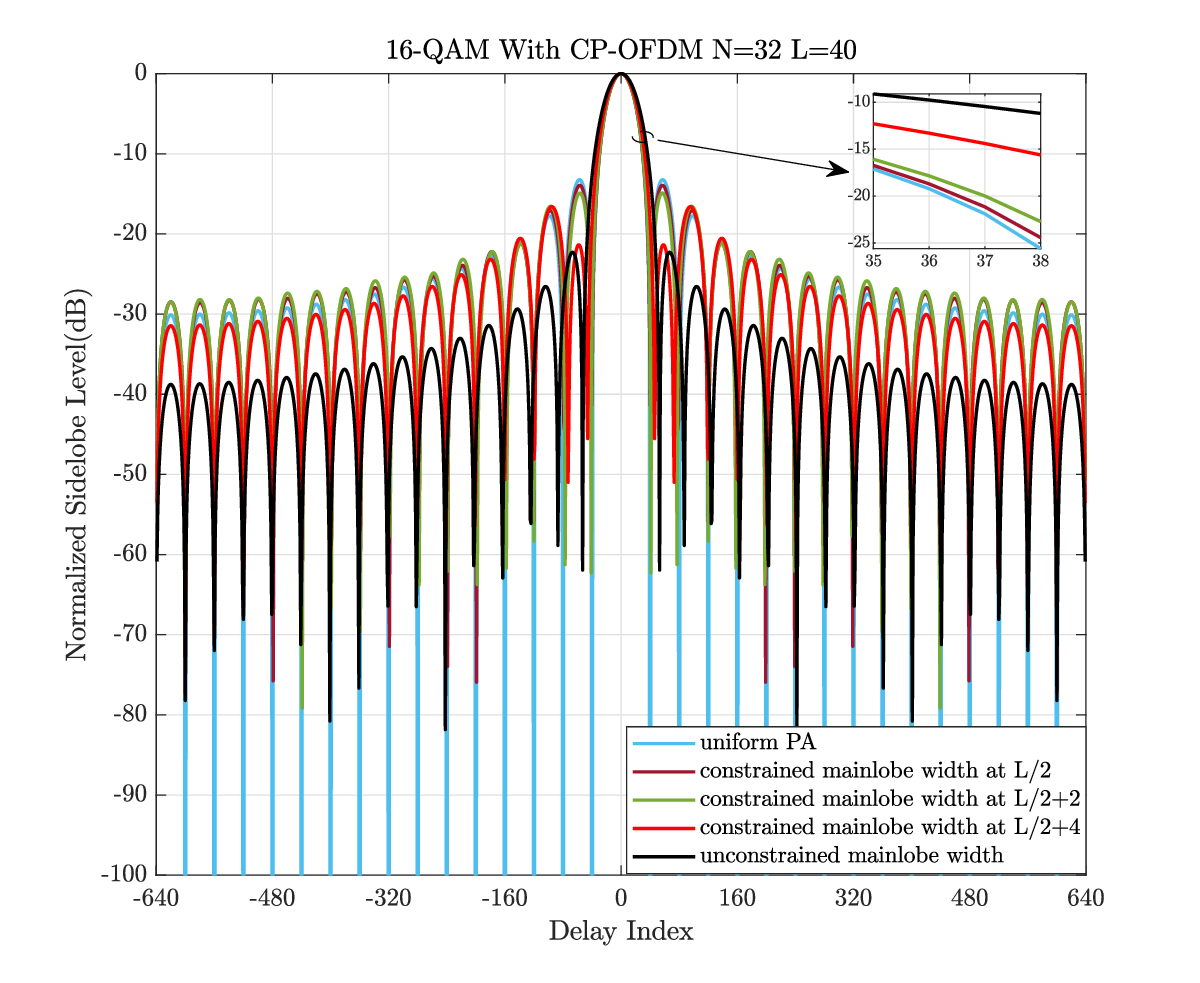} \vspace{-5mm}
	\caption{The P-ACF of 16-PSK under OFDM signaling and different power allocation schemes.}
    \label{fig: oversampling P-ACF compare limit width PSK}
\end{figure} 

Combining Figs. \ref{fig: oversampling P compare limit width}, \ref{fig: oversampling P-ACF compare limit width QAM}, and \ref{fig: oversampling P-ACF compare limit width PSK}, we observe that as the PA scheme approaches the ``uniform PA'' scheme, the mainlobe width becomes narrower, yet the sidelobe level increases. This trend holds for both 16-QAM and 16-PSK constellations. These results suggest a tradeoff between mainlobe width and sidelobe level, allowing for a flexible selection of PA schemes based on specific performance requirements.

\section{Conclusion}\label{Sec_VII}
In this paper, we investigated optimal PA for OFDM-based random ISAC signals. Our findings demonstrated that the uniform PA achieves the global minimum EISL for both periodic and aperiodic autocorrelation functions for all constellations. Additionally, we proved that this scheme achieves the lowest sidelobe at every P-ACF lag. For frequency zero-padding scenarios, where uniform PA is suboptimal, we employed efficient optimization algorithms to design tailored PA schemes that balance the trade-off between sidelobe suppression and mainlobe width. The effectiveness of our theoretical framework was further validated through comprehensive simulations. Future research should explore the optimal signal basis and PA schemes for two-dimensional ambiguity function (2D AF) design.

\appendices
\section{Proof of Proposition 1} \label{Proof of Proposition1}
By referring to \eqref{rk with power allocation-IDFT 2}, the squared P-ACF with power allocation can be expressed as
    \begin{align}\label{AF with power allocation-IDFT 2}
         \left| \widetilde{r}_k \right|^2 \nonumber& = \left( \mathbf{c}^H \mathbf{F}_N \mathbf{J}_k \mathbf{F}_N^H\mathbf{c} \right) \left( \mathbf{c}^H \mathbf{F}_N \mathbf{J}_k \mathbf{F}_N^H\mathbf{c} \right)^* \\
                  \nonumber& = \sum_{n=1}^N \left|c_n\right|^2 e^\frac{j2\pi k(n-1)}{N} \sum_{m=1}^N \left|c_m\right|^2 e^\frac{-j2\pi k(m-1)}{N} \\
                  \nonumber& = \sum_{n=1}^N \sum_{m=1}^N \left|c_n\right|^2\left|c_m\right|^2 e^\frac{-j2\pi k(m-n)}{N} \\
                  & = \sum_{n=1}^N \sum_{m=1}^N P_n P_m \left|s_n\right|^2\left|s_m\right|^2 e^\frac{-j2\pi k(m-n)}{N}.     
    \end{align}
Correspondingly, the average squared P-ACF with PA is 
    \begin{equation}\label{EAF with power allocation-IDFT}
    \begin{split}
         \mathbb{E}(|\widetilde{r}_k|^2)
        & = \mu_4 \sum_{i=1}^N P_i^2 + \sum_{n=1}^N \sum_{\substack{m=1 \\ m\ne n}}^N P_n P_m e^\frac{-j2\pi k(m-n)}{N}\\
        & = (\mu_4-1) \sum_{i=1}^N P_i^2 + \sum_{n=1}^N \sum_{m=1}^N P_n P_m e^\frac{-j2\pi k(m-n)}{N}\\
        & = (\mu_4-1) \sum_{i=1}^N P_i^2 + \left| \sum_{n=1}^N P_n e^\frac{j2\pi kn}{N} \right|^2.
    \end{split}        
    \end{equation}

\section{Proof of Corollary1} \label{Proof of Proposition2}
According to \eqref{EAF with power allocation-IDFT 1}, the EISL of P-ACF with PA can be expressed as
    \begin{equation} \label{EISL OF PACF power}
    \begin{split}
        &\sum_{k=1}^{N-1} \mathbb{E} (|\widetilde{r}_k |^2) 
        = \sum_{k=0}^{N-1} \mathbb{E}(|\widetilde{r}_k|^2)  - \mathbb{E}(|\widetilde{r}_0|^2) \\
        &= \left(N-1\right) \left(\mu_4-1\right) \sum_{i=1}^N P_i^2 
        + \sum_{k=0}^{N-1} \left(\left| \sum_{n=1}^N P_n e^\frac{j2\pi kn}{N} \right|^2 \right) - N^2 .\\
    \end{split}        
    \end{equation}  
Using Parseval's theorem yields
    \begin{equation}
        \sum_{k=0}^{N-1} \left( \left| \sum_{n=1}^N P_n e^\frac{j2\pi kn}{N} \right|^2 \right) =  N \sum_{n=1}^N P_n^2 ,
    \end{equation}
we get \eqref{EAF with power allocation-IDFT sibe}.

\section{Proof of proposition 2} \label{Proof of proposition 3}
According to \eqref{rk with power allocation-IDFT A-ACF 4}
    \begin{align}
        \left| {r}_k \right|^2 
        \nonumber& = \sum_{n=1}^{2N} \left|\mathbf{v}_n^H \mathbf{s} \right|^2 e^{\frac{-j2\pi k(n-1)}{2N}}  \sum_{m=1}^{2N} \left| \mathbf{v}_m^H \mathbf{s} \right|^2 e^{\frac{j2\pi k(m-1)}{2N}} \\
        & = \sum_{n=1}^{2N} \sum_{m=1}^{2N} \left|\mathbf{v}_n^H \mathbf{s}\right|^2 \left|\mathbf{v}_m^H \mathbf{s} \right|^2 e^{\frac{-j2\pi k(n-m)}{2N}}.
    \end{align}
Expanding $\left|\mathbf{v}_n^H \mathbf{s} \right|^2$ yields
    \begin{align}
        |\mathbf{v}_n^H \mathbf{s}|^2 \nonumber& = \mathbf{v}_n^H \mathbf{s} \mathbf{s}^H \mathbf{v}_n = (\mathbf{v}_n^T \otimes \mathbf{v}_n^H) \operatorname{vec}(\mathbf{s}\mathbf{s}^H) \\
        & = (\mathbf{v}_n^T \otimes \mathbf{v}_n^H) \widetilde{\mathbf{s}} = \widetilde{\mathbf{s}}^H (\mathbf{v}_n^* \otimes \mathbf{v}_n). 
    \end{align}
Therefore 
    \begin{equation} \label{Erk2 power allocation A-ACF P-ACF}
    \begin{split}
        \mathbb{E}(|{r}_k|^2) & = \sum_{n=1}^{2N} \sum_{m=1}^{2N} (\mathbf{v}_n^T \otimes \mathbf{v}_n^H) \widetilde{\mathbf{s}} \widetilde{\mathbf{s}}^H (\mathbf{v}_m^* \otimes \mathbf{v}_m) e^{\frac{-j2\pi k(n-m)}{2N}} \\
        & = \sum_{n=1}^{2N} \sum_{m=1}^{2N} (\mathbf{v}_n^T \otimes \mathbf{v}_n^H) \mathbf{S} (\mathbf{v}_m^* \otimes \mathbf{v}_m) e^{\frac{-j2\pi k(n-m)}{2N}},
    \end{split}
    \end{equation}
where $\mathbf{S} = \mathbb{E} ({\widetilde{\mathbf{s}} \widetilde{\mathbf{s}}^H })$. To simplify \eqref{Erk2 power allocation A-ACF P-ACF}, $\mathbf{S}$ can be decomposed \cite{liu2024ofdmachieveslowestranging} as 
    \begin{equation} \label{S}
        \mathbf{S} = \mathbf{I}_{N^2} + \mathbf{S}_1 + \mathbf{S}_2 ,
    \end{equation}
where 
    \begin{align}
        \mathbf{S}_1 &= \mathrm{Diag} \left( \left[\mu_4 - 2, \mathbf{0}_N^T, \mu_4 - 2, \mathbf{0}_N^T,..., \mu_4 - 2 \right]^T \right), \\
        \mathbf{S}_2 &= \left[\mathbf{d}, \mathbf{0}_{N^2 \times N},\mathbf{d},..., \mathbf{d}, \mathbf{0}_{N^2 \times N}, \mathbf{d} \right],
    \end{align}    
with $\mathbf{0}_{N^2\times N}$ being the all-zero matrix of size $N ^2 \times N$, and 
    \begin{equation}
        \mathbf{d} = \left[ 1,\mathbf{0}_N^T,1,...,1,\mathbf{0}_N^T,1 \right]^T.
    \end{equation}
Plugging \eqref{S} into \eqref{Erk2 power allocation A-ACF P-ACF}, we can get
    \begin{align}
       &(\mathbf{v}_n^T \otimes \mathbf{v}_n^H) \mathbf{I}_{N^2} (\mathbf{v}_m^* \otimes \mathbf{v}_m)
        =\left| \mathbf{v}_n^T \mathbf{v}_m^* \right|^2, \label{part1}\\ \nonumber
       &(\mathbf{v}_n^T \otimes \mathbf{v}_n^H) \mathbf{S}_1 (\mathbf{v}_m^* \otimes \mathbf{v}_m) = (\mu_4 - 2) \sum_{p=1}^{N} \left|v_{p,n}\right|^2 \left|v_{p,m}\right|^2\\
        & = (\mu_4 - 2) \| \mathbf{v}_n \odot \mathbf{v}_m \|_{2}^{2} ,\label{part2}\\ \nonumber
       &(\mathbf{v}_n^T \otimes \mathbf{v}_n^H) \mathbf{S}_2 (\mathbf{v}_m^* \otimes \mathbf{v}_m) 
       = \sum_{p=1}^{N} \left|v_{p,n}\right|^2 \sum_{p=1}^{N} \left|v_{p,m}\right|^2 \\ 
       &= \| \mathbf{v}_n \|_{2}^{2} \| \mathbf{v}_m \|_{2}^{2} .\label{part3}
    \end{align}    
According to \eqref{Erk2 power allocation A-ACF P-ACF} - \eqref{part3}, the average mainlobe is
    \begin{equation} \label{mainlobe of power allocation A-ACF P-ACF 2}
    \begin{split}
         \mathbb{E}(|{r}_0|^2) = (\mu_4 -1)\sum_{i=1}^{N}P_{i}^{2} + N^2.
    \end{split}
    \end{equation}

Corresponding to equations \eqref{part1} - \eqref{part3},
$\sum_{k=0}^{2N-1} \mathbb{E}(|{r}_k|^2)$ also can be divided into three terms, as shown in \eqref{Erk2 power allocation A-ACF P-ACF 2}.
\begin{figure*}
    \begin{align} \label{Erk2 power allocation A-ACF P-ACF 2}
        \sum_{k=0}^{2N-1} \mathbb{E}(| r_k |^2) 
        = \sum_{k=0}^{2N-1} \left\{ \sum_{n,m=1}^{2N} | \mathbf{v}_n^T \mathbf{v}_m^* |^2  e^{\frac{-j2\pi k(n-m)}{2N}} 
         +  \sum_{n,m=1}^{2N} (\mu_4 - 2) || \mathbf{v}_n \odot \mathbf{v}_m ||_{2}^{2} e^{\frac{-j2\pi k(n-m)}{2N}} 
         +   \sum_{n,m=1}^{2N} || \mathbf{v}_n ||_{2}^{2} || \mathbf{v}_m ||_{2}^{2}  e^{\frac{-j2\pi k(n-m)}{2N}} \right\}
    \end{align}    
    \hrulefill 
\end{figure*}

\begin{itemize}
    \item 
 The first term of \eqref{Erk2 power allocation A-ACF P-ACF 2} is
    \begin{align}\label{part1 1}
        \nonumber&\sum_{k=0}^{2N-1} \left\{ \sum_{n=1}^{2N} \sum_{m=1}^{2N} \left| \mathbf{v}_n^T \mathbf{v}_m^* \right|^2  e^{\frac{-j2\pi k(n-m)}{2N}} \right\} \\
        \nonumber& = \sum_{n=1}^{2N} \sum_{m=1}^{2N} \left\{ \left| \mathbf{v}_n^T \mathbf{v}_m^* \right|^2  \sum_{k=0}^{2N-1} e^{\frac{-j2\pi k(n-m)}{2N}} \right\} \\    
        & = 2N \sum_{n=1}^{2N} { |\mathbf{v}_n|^4} 
          = 2N \sum_{n=1}^{2N} { \|\mathbf{v}_n\|_{2}^{4}}.
    \end{align}   
    
\item     
The second term of \eqref{Erk2 power allocation A-ACF P-ACF 2} is
    \begin{equation} \label{part2 1}
    \begin{split}
        &\sum_{k=0}^{2N-1} \left\{ \sum_{n=1}^{2N} \sum_{m=1}^{2N} (\mu_4 - 2) \| \mathbf{v}_n \odot \mathbf{v}_m \|_{2}^{2}  e^{\frac{-j2\pi k(n-m)}{2N}} \right\} \\
        & = (\mu_4 - 2) \sum_{k=0}^{2N-1} \left\{ \sum_{n=1}^{2N} \sum_{m=1}^{2N} \| \mathbf{v}_n \odot \mathbf{v}_m \|_{2}^{2}  e^{\frac{-j2\pi k(n-m)}{2N}} \right\} .\\
    \end{split}
    \end{equation}    
We define a vector $\mathbf{b}_k$
    \begin{align} \label{vector bk}
        &\mathbf{b}_k = \left[b_{1,k},b_{2,k},...,b_{N,k}\right]^{T} ,\\
        &b_{p,k} = \sum_{n=1}^{2N} \left|v_{p,n}\right|^2 e^{-j \frac{2\pi k (n-1)}{2N}}, p = 1,2,...,N.
    \end{align}
Using Parseval’s theorem yields
    \begin{align} 
        &\frac{1}{2N} \sum_{k=0}^{2N-1} \left|b_{p,k}\right|^2 = \sum_{n=1}^{2N} \left|v_{p,n}\right|^4, \label{b_p_k ^2 sum}
    \end{align}
and the square of the $\ell_2$ norm of vector $\mathbf{b}_k$ is 
    \begin{align} \label{vector bk^2}
        \| \mathbf{b}_k \|_{2}^{2} \nonumber& = \sum_{p=1}^{N}  \sum_{n=1}^{2N} \sum_{m=1}^{2N} |v_{p,n}|^2 |v_{p,m}|^2 e^{-j \frac{2\pi k (n-m)}{2N}} \\ 
        & = \sum_{n=1}^{2N} \sum_{m=1}^{2N} \| \mathbf{v}_n \odot \mathbf{v}_m \|_{2}^{2} e^{-j \frac{2\pi k (n-m)}{2N}}.
    \end{align}
By combining \eqref{b_p_k ^2 sum} and \eqref{vector bk^2}, we obtain
    \begin{align} \label{vector bk^2 2}
        \nonumber& \sum_{k=0}^{2N-1} \left\{ \sum_{n=1}^{2N} \sum_{m=1}^{2N} \| \mathbf{v}_n \odot \mathbf{v}_m \|_{2}^{2}  e^{\frac{-j2\pi k(n-m)}{2N}} \right\} \\
        & = \sum_{k=0}^{2N-1} \| \mathbf{b}_k \|_{2}^{2} 
        = 2N \sum_{p=1}^{N} \sum_{n=1}^{2N} |v_{p,n}|^4 
        = 2N \|\mathbf{V}\|_{4}^{4}.   
    \end{align}
Plugging \eqref{vector bk^2 2} into \eqref{part2 1}, we arrive at
    \begin{equation} \label{part2 1-1}
        2N(\mu_4 - 2) \|\mathbf{V}\|_{4}^{4}.
    \end{equation}
    
\item  The third term of \eqref{Erk2 power allocation A-ACF P-ACF 2} is
    \begin{align}\label{part3 1}
        \nonumber&\sum_{k=0}^{2N-1} \left\{ \sum_{n=1}^{2N} \sum_{m=1}^{2N} \| \mathbf{v}_n \|_{2}^{2} \| \mathbf{v}_m \|_{2}^{2}  e^{\frac{-j2\pi k(n-m)}{2N}} \right\} \\
        \nonumber& = \sum_{k=0}^{2N-1} 
        \left| \sum_{n=1}^{2N}\|\mathbf{v}_n\|_{2}^{2} e^{\frac{-j2\pi k(n-1))}{2N}} \right|^2\\ 
        &= 2N \sum_{n=1}^{2N} \|\mathbf{v}_n\|_{2}^{4}.
    \end{align}
    
\end{itemize}
Combining \eqref{mainlobe of power allocation A-ACF P-ACF 2}, \eqref{part1 1}, \eqref{part2 1-1} and \eqref{part3 1}, the EISL of the A-ACF is shown in \eqref{EISL of A-ACF power allocation}.


\section{Proof of theorem 3} \label{Proof of theorem 3}
According to \eqref{maiblobe of A-ACF power allocation} and \eqref{EISL of A-ACF power allocation}, the normalized EISL of A-ACF with PA is 
    \begin{equation} \label{normalized EISL of A-ACF with power}
    \begin{split}
        \frac{\frac{1}{2} \sum_{k=1}^{2N-1} \mathbb{E}(|r_k|^2)}{\mathbb{E}(|r_0|^2)} = 
        N\frac{ \sum_{n=1}^{2N} \left[ (\mu_4 - 2)\|\mathbf{v}_n\|_{4}^{4} + 2\|\mathbf{v}_n\|_{2}^{4} \right]}{(\mu_4-1) \sum_{i=1}^{N} P_{i}^2 + N^2} - \frac{1}{2}.
    \end{split}    
    \end{equation}

By observing the structure of $\widetilde{\mathbf{F}}_{2N}$, we can extract all the odd rows and reorganize them into a matrix $\widetilde{\mathbf{F}}_{2N,o}$, while all the even rows can be reorganized as $\widetilde{\mathbf{F}}_{2N,e}$.
    \begin{align} 
        &\widetilde{\mathbf{F}}_{2N,o} = \frac{1}{\sqrt{2}} \mathbf{F}_N , \qquad
        \widetilde{\mathbf{F}}_{2N,e} = \frac{1}{\sqrt{2}} \mathbf{F}_N \mathbf{D}_{\frac{1}{2}} ,
    \end{align}
where $\mathbf{D}_{\alpha}$ is a diagonal matrix with its \textit{n}-th diagonal entry begin ${e^{-j \frac{2\pi (n-1)}{N} \cdot \alpha}}$.

We set $\mathbf{W} = \mathbf{F}_N \widetilde{\mathbf{F}}_{2N}^{H} = \left[ \mathbf{w}_1, \mathbf{w}_2,...,\mathbf{w}_{2N} \right]$, so $\mathbf{V} = \mathbf{B} \mathbf{W} = \left[ \mathbf{B} \mathbf{w}_1, \mathbf{B} \mathbf{w}_2, ... , \mathbf{B} \mathbf{w}_{2N} \right]$. Because of the specificity of the $\widetilde{\mathbf{F}}_{2N}$, we can find that all odd-numbered columns of $\mathbf{W}$ can form an identity matrix, and all even-numbered columns can form a cyclic matrix.
    \begin{align}
        &\mathbf{W}1 = \left[ \mathbf{w}_1, \mathbf{w}_3,...,\mathbf{w}_{2N-1} \right] = \mathbf{F}_N \widetilde{\mathbf{F}}_{2N,o}^{H} = \frac{1}{\sqrt{2}} \mathbf{I}_N, \label{W1} \\
        &\mathbf{W}2 = \left[ \mathbf{w}_2, \mathbf{w}_4,...,\mathbf{w}_{2N} \right] 
        = \mathbf{F}_N \widetilde{\mathbf{F}}_{2N,e}^{H} = \frac{1}{\sqrt{2}} \mathbf{F}_N \mathbf{D}_{\frac{1}{2}}^{H} \mathbf{F}_N^H. \label{W2}
    \end{align}
Similarly, the first term of the RHS of \eqref{normalized EISL of A-ACF with power} can be divided according to odd and even columns. Accordingly, the normalized EISL of A-ACF \eqref{normalized EISL of A-ACF with power} can be rewritten as \eqref{normalized EISL of A-ACF with power 2}.
\begin{figure*}
    \begin{align} \label{normalized EISL of A-ACF with power 2}
        \frac{\frac{1}{2} \sum_{k=1}^{2N-1} \mathbb{E}(|r_k|^2)}{\mathbb{E}(|r_0|^2)} 
        = N\frac{ \sum_{n=1}^{N} \left[ (\mu_4 - 2) \|\mathbf{v}_{2n-1}\|_{4}^{4} + 2\|\mathbf{v}_{2n-1}\|_{2}^{4} \right]}{(\mu_4-1) \sum_{i=1}^{N} P_{i}^2 + N^2} 
        + N\frac{ \sum_{n=1}^{N} \left[ (\mu_4 - 2)\|\mathbf{v}_{2n}\|_{4}^{4} + 2\|\mathbf{v}_{2n}\|_{2}^{4} \right]}{(\mu_4-1) \sum_{i=1}^{N} P_{i}^2 + N^2} - \frac{1}{2} 
    \end{align}    
\end{figure*}

According to \eqref{B} and \eqref{W1}, we can get
    \begin{equation} \label{norm forth power}
        \|\mathbf{v}_{2n-1}\|_{4}^{4} = \|\mathbf{v}_{2n-1}\|_{2}^{4} = \frac{1}{4} P_{n}^{2}, n = 1,2, ... ,N.
    \end{equation}
Plugging \eqref{norm forth power} into the first term of \eqref{normalized EISL of A-ACF with power 2}, this term can be reduced to \eqref{normalized EISL of A-ACF with power 3}. When the value of $\sum_{n=1}^{N} P_{n}^{2}$ is minimized, the value of \eqref{normalized EISL of A-ACF with power 3} is also minimized. According to \eqref{normalize EAF with power allocation-IDFT sibe} and \eqref{value of P}, only when $P_1 = P_2 = ... = P_N = 1$, \eqref{normalized EISL of A-ACF with power 3} reaches the minimum. 
\begin{figure*}
    \begin{align} \label{normalized EISL of A-ACF with power 3}
        \frac{ \sum_{n=1}^{N} \left[ (\mu_4 - 2) \|\mathbf{v}_{2n-1}\|_{4}^{4} + 2\|\mathbf{v}_{2n-1}\|_{2}^{4} \right]}{(\mu_4-1) \sum_{i=1}^{N} P_{i}^2 + N^2} 
        = \frac{\mu_4}{4} \frac{\sum_{n=1}^{N} P_{n}^{2}}{(\mu_4-1) \sum_{i=1}^{N} P_{i}^2 + N^2} 
        =  \frac{\mu_4}{4} \frac{1}{(\mu_4-1) + \frac{N^2}{\sum_{n=1}^{N} P_{n}^{2}}}
    \end{align}    
\end{figure*}

Based on the conditions that $\mu_4 \ge 1$ and $||\mathbf{v}_{n}||_{2}^{4} \ge ||\mathbf{v}_{n}||_{4}^{4}$, the Cauchy-Schwarz Inequality can be used in the second term of \eqref{normalized EISL of A-ACF with power 2}, as is shown in \eqref{normalized EISL of A-ACF with power 4}. Only when every $\sqrt{ (\mu_4 - 2)||\mathbf{v}_{2n}||_{4}^{4} + 2||\mathbf{v}_{2n}||_{2}^{4} }$  $(n = 1,2,...,N)$ is equal, the \eqref{normalized EISL of A-ACF with power 4} reaches the minimum. 
According to \eqref{W2}, when $P_1 = P_2 = ... = P_N = 1$, the conditions for equality are met. 

In conclusion, when $P_1 = P_2 = ... = P_N = 1$, the normalized EISL achieves the global minimum.

\begin{figure*}[ht]
    \begin{align}\label{normalized EISL of A-ACF with power 4}
        N \frac{ \sum_{n=1}^{N} \left[ (\mu_4 - 2) \|\mathbf{v}_{2n}\|_{4}^{4} + 2\|\mathbf{v}_{2n}\|_{2}^{4} \right]}{(\mu_4-1) \sum_{i=1}^{N} P_{i}^2 + N^2} 
        \nonumber &= \sum_{n=1}^{N} \frac{1}{(\mu_4-1) \sum_{i=1}^{N} P_{i}^2 + N^2}
         \times \sum_{n=1}^{N} \left[ (\mu_4 - 2)\|\mathbf{v}_{2n}\|_{4}^{4} + 2\|\mathbf{v}_{2n}\|_{2}^{4} \right] \\
        &\ge \left( \sum_{n=1}^{N} \frac{ \sqrt{(\mu_4 - 2) \|\mathbf{v}_{2n}\|_{4}^{4} + 2\|\mathbf{v}_{2n}\|_{2}^{4} }}{\sqrt{ (\mu_4-1) \sum_{i=1}^{N} P_{i}^2 + N^2 }} \right)^2 
    \end{align}
    \hrulefill 
\end{figure*}

\section{Proof of proposition 3} \label{proof of proposition 4}
According to \eqref{rk with power allocation and zero padding}, the sidelobe of $\hat{r}_k$ is 
    \begin{equation}
    \begin{split}
        |\hat{r}_k|^2 &= \left| \sum_{n=1}^N {P_n}|s_n|^2 e^\frac{j2\pi k(n-1)}{NL} \right|^2 \\
        &= \sum_{n=1}^N \sum_{m=1}^N P_n P_m |s_n|^2 |s_m|^2 e^\frac{j2\pi k(n-m)}{NL}.
    \end{split}
    \end{equation}
The average sidelobe of P-ACF with PA and frequency zero-padding is 
    \begin{equation}
    \begin{split}
         \mathbb{E} (|\hat{r}_k|^2) 
        &= \mu_4 \sum_{i=1}^N P_i^2 + \sum_{n=1}^N \sum_{\substack{m=1 \\ m\neq n}}^N P_n P_m e^\frac{j2\pi k(n-m)}{NL} \\
        &= (\mu_4-1) \sum_{i=1}^N P_i^2 + \left| \sum_{n=1}^N P_n e^\frac{j2\pi kn}{NL} \right|^2.
        \end{split}
    \end{equation}
Taking $k=0$ into the above equation, the average mainlobe of the P-ACF with PA and frequency zero-padding is
    \begin{equation}
        \mathbb{E} (|\hat{r}_0|^2) = (\mu_4-1) \sum_{i=1}^N P_i^2 + N^2.
    \end{equation}

The $\text{EISL}_{\text{ZP}}$ of the P-ACF with PA and frequency zero-padding is 
    \begin{equation} \label{EISL P-ACF power allocation zero-padding}
    \begin{aligned}
        &\sum_{k=L}^{\frac{NL}{2}-1} \mathbb{E} (|\hat{r}_k|^2)  \\
        &= \sum_{k=L}^{\frac{NL}{2}-1} \left[ (\mu_4-1) \sum_{i=1}^N P_i^2 + \left| \sum_{n=1}^N P_n e^\frac{j2\pi kn}{NL} \right|^2 \right] \\
        &= (\frac{NL}{2} - L) (\mu_4-1) \sum_{i=1}^N P_i^2 + \sum_{k=L}^{\frac{NL}{2}-1} \left| \sum_{n=1}^N P_n e^\frac{j2\pi kn}{NL} \right|^2.
    \end{aligned}        
    \end{equation}
For the convenience of writing, we define $\mathbf{g}_k$ and $\mathbf{G}$ as
    \begin{align} 
       \mathbf{g}_k &= 
           \left[ e^\frac{-j2\pi k}{NL},
            e^\frac{-j2\pi 2k}{NL},
            . 
            . 
            . ,
            e^\frac{-j2\pi Nk}{NL} \right]^{T} \label{new vector g_k} ,\\
       \mathbf{G} &= 
        \left [
            \mathbf{g}_L, \mathbf{g}_{L+1}, ... ,\mathbf{g}_{NL/2-1}
        \right ],\label{new G}
    \end{align}
and the size of $\mathbf{G}$ is $N \times (\frac{NL}{2} - L)$.
Plugging \eqref{new vector g_k} and \eqref{new G} into \eqref{EISL P-ACF power allocation zero-padding}, we can get 
    \begin{equation}
    \begin{aligned}
    \sum_{k=L}^{\frac{NL}{2}-1} \left| \sum_{n=1}^N P_n e^\frac{j2\pi kn}{NL} \right|^2 = \sum_{k=L}^{\frac{NL}{2}-1} \left| \mathbf{g}_k^H \mathbf{p} \right|^2
    = \left\| \mathbf{G}^H \mathbf{p} \right\|_2^2,
    \end{aligned}
    \end{equation}
where $ \mathbf{p} = [P_1,P_2,...,P_N]^T$. Correspondingly, \eqref{EISL P-ACF power allocation zero-padding} can be rewritten as
    \begin{equation}
        \sum_{k=L}^{\frac{NL}{2}-1} \mathbb{E} (|\hat{r}_k|^2) = (\frac{NL}{2}-L)(\mu_4-1) \| \mathbf{p} \|_2^2 + \left\| \mathbf{G}^H \mathbf{p} \right\|_2^2.
    \end{equation}

\bibliographystyle{IEEEtran}
\bibliography{reference_zy}

\end{document}